\documentclass[letterpaper,journal]{IEEEtran}
\usepackage{amsmath,amsfonts,amssymb}
\usepackage{algorithmic}
\usepackage{algorithm}
\usepackage{array}
\usepackage{multirow}
\usepackage{booktabs}
\usepackage[caption=false,font=normalsize,labelfont=sf,textfont=sf]{subfig}
\usepackage{textcomp}
\usepackage{stfloats}
\usepackage{url}
\usepackage{bm}
\usepackage{xcolor}
\usepackage{verbatim}
\usepackage{graphicx}
\usepackage{cite}
\usepackage[capitalise]{cleveref}
\usepackage{siunitx}
\usepackage{amsthm}
\usepackage{mathtools}
\usepackage{enumitem}
\usepackage{tikz}
\usepackage{fancyhdr}
\newcommand{\copyrightstatement}{%
This work has been submitted to the IEEE for possible publication. Copyright may be transferred without notice, after which\\
this version may no longer be accessible.%
}

\fancypagestyle{paperheader}{%
    \fancyhf{}
    \fancyhead[C]{%
        \footnotesize
        \parbox{0.86\textwidth}{\centering \copyrightstatement}
    }
    \fancyhead[R]{\footnotesize\thepage}

}

\usepackage{makecell}
\usepackage{pifont}
\usepackage{placeins}
\usepackage{flushend}
\usetikzlibrary{
    arrows.meta,
    calc,
    decorations.pathreplacing,
    shapes.misc,
    shapes.geometric
}

\newtheorem{theorem}{Theorem}
\newtheorem{lemma}{Lemma}

\theoremstyle{definition}

\DeclareCaptionFont{probehalf}{\fontsize{5.5pt}{6.5pt}\selectfont}
\begin{document}
\onecolumn
\thispagestyle{empty}

\vspace*{0.8in}

\begin{center}
    {\LARGE\bfseries Copyright Statements\par}
\end{center}

\vspace*{0.7in}

{\noindent\Large
This work has been submitted to the IEEE for possible publication. Copyright may be
transferred without notice, after which this version may no longer be accessible.
\par}

\clearpage
\setcounter{page}{1}
\twocolumn

\title{The Intrinsic Inertia, Droop, and Nonlinear Stability of Observer-Synchronized Grid-Forming Control}
\author{Aizuo Chen,
    Xiangyu Meng,~\IEEEmembership{Member,~IEEE}, Haozhe Wang, Boxin Liu, Yue Zhu,~\IEEEmembership{Member,~IEEE}}
\markboth{}%
{Chen \MakeLowercase{\textit{*et al.*}}: Nonlinear Stability and Droop Characteristics of Observer-Synchronized Grid-Forming Control}
\maketitle
\thispagestyle{paperheader}
\pagestyle{paperheader}

\begin{abstract}
    \textcolor{black}{Observer-synchronized (OBS) grid-forming
        (GFM) control, in which grid synchronization,
        active-power regulation, voltage-magnitude regulation, and current
        limiting are integrated into a unified observer-based structure, was proposed in recent years as an alternative to conventional GFM control such as droop control.
        Existing research on OBS GFM mainly focuses on its linearized modelling, small-signal stability, and parameter sensitivity,
        while a nonlinear analysis, which characterizes its inertia, droop characteristics, and nonlinear convergence, is currently absent. This paper
    develops a systematic nonlinear analysis of OBS control. First, a second-order nonlinear autonomous model is derived under the quasi-steady-state assumption. 
    The intrinsic inertia is characterized through an exact nonlinear swing-equation representation, while the nonlinear droop characteristics are derived from the steady-state relations. These results provide guidelines for controller parameterization. Furthermore, the equilibrium points and nonlinear stability properties are investigated using the qualitative theory of differential equations, giving explicit criteria to guarantee almost global asymptotic stability of the desired equilibriums.} The analysis is validated
    through detailed electromagnetic transient simulations and hardware experiments.

\end{abstract}


\section{Introduction}
\label{sec:introduction}

\IEEEPARstart{T}{he} increasing penetration of renewable energy sources, such as photovoltaic and wind generation, is reshaping modern power systems from synchronous-generator (SG)-dominated toward those dominated by inverter-based resources (IBRs) \cite{Gu2023Power}. As SGs are displaced by power electronic inverters, declining system inertia and strength pose challenges to frequency and voltage stability \cite{Hatziargyriou2021Definition,Zhu2024IMR}. Grid-forming (GFM) inverters have therefore attracted considerable attention for their ability to support frequency and voltage stability \cite{Bahrani2024GFMLandscape,Rathnayake2021GFMModeling,Li2022Revisiting}.

\color{black}
Several distinct synchronization mechanisms have been developed for GFM inverters. Droop control achieves synchronization through active-power--frequency ($p$--$\omega$) coupling \cite{Chandorkar1993Parallel,Pogaku2007Modeling,Du2020Comparative}. Virtual synchronous generator (VSG) control emulates the swing dynamics of synchronous machines \cite{Zhong2011Synchronverter,Li2023VSGInertia}. In contrast, virtual oscillator control (VOC) and its variants realize synchronization through nonlinear oscillator dynamics \cite{Lu2022VOCState,Torres2015Synchronization,Colombino2019GlobalSync,Seo2019DispatchableVOC,Kong2022PBVOC,dvoc-1,He2024ComplexDroop}.

Observers have also been widely adopted in grid-connected inverter control. For example, a state observer is used in \cite{Kukkola2015Observer} to improve the dynamic performance and suppress resonance in a grid-following (GFL) inverter. In \cite{Heo2025DOBVI}, a disturbance observer is introduced to improve the virtual-impedance implementation of VSG control, while \cite{Zhang2025Observer} employs an extended state observer in a Droop-controlled GFM inverter for sensor reduction. In these approaches, the observer serves as an auxiliary estimation or compensation mechanism within an existing synchronization or control framework, rather than constituting the synchronization mechanism itself.

\color{black}
\textcolor{black}{In contrast, an observer-synchronized (OBS) GFM
    control is proposed in \cite{obGFM}, integrating synchronization, active-power
    regulation, voltage-magnitude regulation, and current limiting into a common
    observer-based structure. This unified structure enables seamless GFL/GFM
    transitions.
    The sensitivity analysis in \cite{obGFM-1} identifies parameter choices
    that preserve robustness against grid-inductance estimation errors.
    To improve robustness against grid-impedance uncertainty,
    \cite{Gao2026ObserverSynchronized} combines observer-based load-angle tracking
    with disturbance compensation, achieving robust synchronization and
    offset-free power tracking. The original and robust OBS schemes also
    exhibit inherent droop, which is verified numerically and experimentally.
    However, these studies do not reveal how intrinsic inertia and nonlinear
    droop emerge from the OBS dynamics or how the controller parameters
    shape these characteristics. They also do not establish explicit
    sufficient parameter conditions for nonlinear convergence.}

\begingroup\color{black}
Several complementary methods can be used to analyze nonlinear convergence.
Equilibrium analysis identifies the existence and location of synchronized
operating points \cite{Jin2022DQFrame,Tang2022ImprovedEAC,He2024ComplexDroop},
but their existence alone does not establish convergence.
Lyapunov and energy-based methods can certify nonlinear stability and estimate
regions of attraction through scalar functions and their sublevel sets
\cite{GFL5,Ge2023VSGTransientBoundary,Fu2021LargeSignal,Kong2022PBVOC,Torres2015Synchronization,Colombino2019GlobalSync}.
However, constructing a suitable Lyapunov function can be challenging, and
the resulting estimates of the region of attraction may be conservative.
Geometric and qualitative analysis relates the structure of the vector field
to trajectory behavior.
For planar autonomous systems, the Poincar\'e--Bendixson theorem
enables convergence analysis without a Lyapunov function
\cite{perko2013differential}.
Geometric analysis reveals coexisting stable equilibria and stable periodic
orbits in a droop-controlled GFM converter \cite{Yang2021Homoclinic},
and shows that bounded trajectories of a second-order dVOC model approach
limit cycles if they do not converge to equilibria \cite{He2024ComplexDroop}.
It also yields sufficient parameter conditions for
almost-global asymptotic stability of GFL inverters based on current control
synchronization \cite{ccs}.

For OBS control, time-scale separation permits a quasi-steady-state
approximation of the electrical network, reducing the synchronization
dynamics to a planar autonomous system in the observer-voltage magnitude
and relative phase. The nonlinear coupling between these states
complicates the construction of a suitable Lyapunov function.
Geometric analysis can instead characterize trajectory evolution directly
from the reduced vector field and establish convergence using the
Poincar\'e--Bendixson theorem, by proving trajectory boundedness and
excluding non-equilibrium limit sets.
\par\endgroup

\color{black}
This paper develops a systematic nonlinear analysis of OBS control, linking
its intrinsic inertia and nonlinear droop characteristics to the controller
parameters and establishing sufficient conditions for almost-global
convergence of the reduced model. The main contributions are summarized as follows:
\textcolor{black}{
\begin{itemize}[leftmargin=*]
    \item[1)] A reduced second-order nonlinear model is derived under the quasi-steady-state assumption for nonlinear convergence analysis. An exact nonlinear swing-equation representation characterizes the intrinsic inertia of OBS control.
    \item[2)] The nonlinear $p$--$\omega$ and voltage--frequency
          characteristics are revealed, and their dependence on the controller
          parameters is established analytically, enabling direct
          parameterization of the OBS controller.
    \item[3)] The equilibrium structure and nonlinear convergence properties of
          the reduced OBS model are rigorously characterized, and explicit
          sufficient conditions for almost-global convergence are established.
    \item [4)] The structural relations of OBS control to other GFM synchronization laws are comprehensively compared, with their primary correspondence uncovered.
\end{itemize}}

\color{black}
The remainder of this paper is organized as follows. Section~\ref{sec:obs_theory} presents the OBS model, intrinsic inertia and droop characteristics, parameterization, and relation to established GFM synchronization laws. Section~\ref{sec:nonlinear_stability} establishes sufficient conditions for almost-global convergence of the reduced OBS model. Section~\ref{sec:validation} provides numerical and electromagnetic transient (EMT) validation, including reference responses with VSG and Droop control. Section~\ref{sec:experimental_validation} presents the experimental validation, and Section~\ref{sec:conclusion} concludes the paper.

\section{Modeling and Theoretical Analysis of OBS Control}
\label{sec:obs_theory}

Unless otherwise specified, all equations are written in per unit with
$\tau=\omega_b t$, and a dot denotes differentiation with respect to $\tau$.

\subsection{Review of OBS Control}
\label{subsec:obs_control_review}

For the OBS control in \cite{obGFM}, the power-invariant $dq$ variables in the
$\omega_b$-rotating reference frame are written in complex form as
$i=i_d+ji_q$, $v=v_d+jv_q$, $\hat v=\hat v_d+j\hat v_q$, and
$\Delta v=\Delta v_d+j\Delta v_q$. Define
$\hat V=|\hat v|$ and the associated unit phasor $n=\hat v/\hat V$.
The current magnitude is $|i|=\sqrt{i_d^2+i_q^2}$.

\color{black}
The observer voltage dynamics are
\begin{equation}
    \dot{\hat v}
    =
    k_o\Delta v
    -(k_o-j)l_o\dot i,
    \label{ob}
\end{equation}
where $k_o$ and $l_o$ are the observer gain and observer inductance,
respectively, as defined in \cite{obGFM}. Explicit differentiation of $i$ is
not required in implementation; see \cite{obGFM}. The estimated active and reactive powers are
$\hat p=\Re(\hat v\overline{i})$ and
$\hat q=\Im(\hat v\overline{i})$. The feedback correction is
\begin{equation}
    \Delta v
    =
    n
    \left[
        k_p(p_r-\hat p)
        +(1-jk_v)(v_r-\hat V)
        \right],
    \label{eq:delta_v}
\end{equation}
where $p_r$ and $v_r$ are the active-power and voltage-magnitude references,
and $k_p,k_v>0$ are controller gains.

\color{black}
Unlike conventional GFM structures that regulate the filter-capacitor voltage,
the original OBS control regulates the inverter-terminal voltage. Owing to
the inverter-side filter impedance, this makes the capacitor voltage sensitive
to current variations caused by external disturbances \cite{dvoc-2}.
To reduce this sensitivity, impedance feedforward is introduced to shift
the effective regulation port to the filter capacitor.
Transient-current feedback is also incorporated to suppress fast current oscillations.
The resulting OBS implementation is shown in Fig.~\ref{controller},
where $z_f=r_f+jl_f$ denotes the inverter-side filter impedance,
with $r_f$ and $l_f$ being the filter resistance and inductance, respectively.

\begin{figure}[!t]
    \centering
    \includegraphics[width=3.4in]{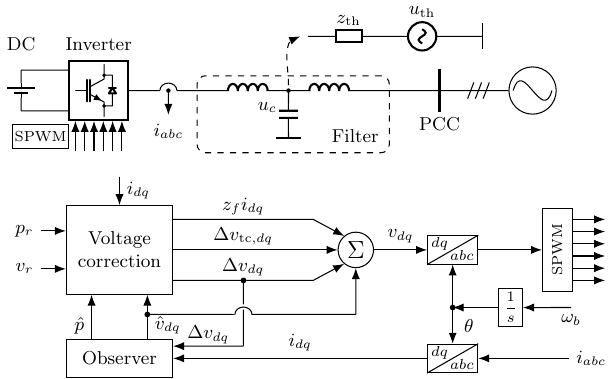}
    \caption{\color{black}OBS control diagram in a single-inverter infinite-bus system.}
    \label{controller}
\end{figure}

With $i_a=i\overline n$, the current low-pass filter and transient-current feedback voltage are
\begin{equation}
    \dot\chi = k_a(i_a-\chi),
    \qquad
    \Delta v_{\mathrm{tc}} = -r_a n(i_a-\chi).
\end{equation}
where $k_a>0$ is the filter bandwidth and $r_a>0$ is the transient-current feedback gain.

The inverter voltage reference is
\begin{equation}
    v
    =
    \hat v+\Delta v+\Delta v_{\mathrm{tc}}
    +z_fi,
    \label{eq:capacitor_voltage_reference}
\end{equation}
Since $\Delta v_{\mathrm{tc}}$ acts only on the inverter-voltage
command, the OBS synchronization law
\eqref{ob}--\eqref{eq:delta_v} remains unchanged.

With ideal voltage tracking, the inverter-side dynamics become
\begin{equation}
    l_f\dot i
    =
    \hat v+\Delta v+\Delta v_{\mathrm{tc}}-u_c.
    \label{eq:capacitor_translation}
\end{equation}
where $u_c$ is the filter-capacitor voltage. Thus, impedance feedforward
shifts the effective regulation port to the filter capacitor.
At synchronized steady state, $\chi=i_a$ and $\Delta v_{\mathrm{tc}}=0$,
so this feedback term does not alter the steady-state relations.

\subsection{Nonlinear System Modeling and Reduction}
\label{subsec:nonlinear_reduction}
The analysis below considers OBS with $l_o=0$ and characterizes the slow synchronization dynamics. The bandwidth $k_a$ is selected such that $\chi$ tracks $i_a$ on a
faster time scale than the OBS synchronization dynamics. Accordingly, the
transient-current feedback is neglected in the reduced nonlinear model under
the approximation $i_a-\chi\simeq0$ and $\Delta v_{\mathrm{tc}}\simeq0$. Let the grid angular frequency be denoted by $\omega_g$ and
define the grid-frequency deviation as $\Delta\omega:=\omega_g-1$.

Transforming \eqref{ob} into the $\omega_g$-rotating reference frame and
retaining the same symbols gives
\begin{equation}
    \dot{\hat v}
    =
    k_o\Delta v-j\Delta\omega\hat v.
    \label{eqd}
\end{equation}

In the grid-synchronous frame, the complete inverter-side current dynamics are
\begin{equation}
    l_f\dot i
    =
    \hat v+\Delta v+\Delta v_{\mathrm{tc}}-u_c
    -j\Delta\omega l_fi.
    \label{eq:capacitor_translation_grid}
\end{equation}

As shown in Fig.~\ref{controller}, let
$i=(u_c-u_{\mathrm{th}})/z_{\mathrm{th}}$
represent the Thevenin equivalent of the external network seen from the
filter-capacitor port. Choosing the phase of $u_{\mathrm{th}}$ as the reference, define
$u_{\mathrm{th}}=U$ and
$z_{\mathrm{th}}+j\Delta\omega l_f=Ze^{j\phi}$.
Combining the network relation with \eqref{eq:capacitor_translation_grid}
under $\dot i=0$ and $\Delta v_{\mathrm{tc}}\simeq0$ yields the reduced model
\begin{equation}
    \label{symp}
    \dot{\hat v}=k_o\Delta v-j\Delta\omega\hat v,
    \qquad
    i=\frac{\hat v+\Delta v-U}{Ze^{j\phi}}.
\end{equation}
For a specified
constant grid-frequency deviation $\Delta\omega$, $U$, $Z$, and $\phi$ are
constants.

In this grid-synchronous frame, write $\hat v=\hat V e^{j\delta}$, where $\delta$ denotes the phase of $\hat v$ relative to the phase reference $u_{\mathrm{th}}$. Multiplication of $\dot{\hat v}$ by $\overline{\hat v}$ gives
\begin{equation}
    \dot{\hat v}\overline{\hat v}
    =
    k_o\hat V
    \left[
        k_p(p_r-\hat p)
        +(1-jk_v)(v_r-\hat V)
        \right]
    -j\Delta\omega \hat V^2.
    \label{eq_vvc}
\end{equation}
Separating the real and imaginary parts of \eqref{eq_vvc} yields the dynamics of $\hat V$ and $\delta$:
\begin{align}
    \dot{\hat V}
     & =
    k_ok_p(p_r-\hat p)
    +k_o(v_r-\hat V),
    \label{dv}
    \\
    \dot\delta
     & =
    -\frac{k_ok_v}{\hat V}(v_r-\hat V)
    -\Delta\omega.
    \label{dtheta}
\end{align}

Let $\omega_o$ denote the angular frequency of $\hat v$ and
$\Delta\omega_o:=\omega_o-1$ represent its deviation from the nominal frequency.
Since $\delta$ is the phase displacement of $\hat v$ relative to the grid-synchronous reference and the derivative is taken with respect to $\tau$, one has $\dot\delta=\Delta\omega_o-\Delta\omega$.
Combining this relation with \eqref{dtheta} gives
\begin{equation}
    \Delta\omega_o
    =
    -\frac{k_ok_v}{\hat V}(v_r-\hat V).
    \label{eq:obs_internal_frequency}
\end{equation}
Substituting $\hat v=\hat V e^{j\delta}$ and \eqref{eq:delta_v} into the algebraic current constraint in \eqref{symp} yields the explicit expression:
\begin{equation}
    i
    =
    \frac{
        e^{j\delta}
        \big(
        \hat V+k_p(p_r-\hat p)
        +(1-jk_v)(v_r-\hat V)
        \big)-U
    }{Ze^{j\phi}},
\end{equation}
Accordingly, the estimated complex power is
\begin{equation}
    \hat v \overline{i}
    =
    \frac{\hat V}{Z}e^{j\phi}
    \left[
    v_r
    +k_p(p_r-\hat p)
    +jk_v(v_r-\hat V)
    -Ue^{j\delta}
    \right].
    \label{eq_apparent}
\end{equation}
Extracting the real and imaginary parts gives
\begin{equation}
    \begin{aligned}
        \hat p(\hat V,\delta)
         & =
        \frac{\hat V}{Z+k_p\hat V\cos\phi}
        \Big[
            (k_pp_r+v_r)\cos\phi
        \\
         & \qquad
            +k_v(\hat V-v_r)\sin\phi
            -U\cos(\delta+\phi)
            \Big],
        \\[1ex]
        \hat q(\hat V,\delta)
         & =
        \frac{\hat V}{Z}
        \Big[
            (-k_p\hat p+k_pp_r+v_r)\sin\phi
        \\
         & \qquad
            +k_v(v_r-\hat V)\cos\phi
            -U\sin(\delta+\phi)
            \Big].
    \end{aligned}
    \label{eq_pq_hat}
\end{equation}
It follows from \eqref{eq_pq_hat} that $\hat p$ and $\hat q$ are nonlinear functions of the state variables $\hat V$ and $\delta$. Substitution into \eqref{dv}--\eqref{dtheta} yields a second-order nonlinear autonomous system, which forms the basis of the subsequent analysis.
\subsection{Intrinsic Inertia and Droop Characteristics}
\label{subsec:intrinsic_characteristics}
\begingroup
\subsubsection{Intrinsic Inertia}
\color{black}
Equation~\eqref{eq:obs_internal_frequency} shows that the OBS frequency is generated by the estimated voltage-magnitude dynamics rather than by an explicitly implemented swing equation. Solving \eqref{eq:obs_internal_frequency} for $\hat V$ gives
\begin{equation}
    \hat V
    =
    \frac{v_r}{1-\Delta\omega_o/(k_ok_v)}.
    \label{eq:V_frequency_manifold}
\end{equation}
The physically meaningful region $\hat V>0$ requires $1-\Delta\omega_o/(k_ok_v)>0$. Differentiating \eqref{eq:V_frequency_manifold} gives
\begin{equation}
    \dot{\hat V}
    =
    \frac{v_r\dot{\Delta\omega}_o}
    {k_ok_v[1-\Delta\omega_o/(k_ok_v)]^2}.
    \label{eq:V_frequency_derivative}
\end{equation}
Substituting \eqref{eq:V_frequency_derivative} into \eqref{dv}
and using \eqref{eq:V_frequency_manifold} yields
\begin{equation}
    M_o(\Delta\omega_o)
    \dot{\Delta\omega}_o
    =
    p_r-\hat p
    -D_o(\Delta\omega_o)\Delta\omega_o,
    \label{eq:nonlinear_swing}
\end{equation}
where
\begin{align}
    M_o(\Delta\omega_o)
     & =
    \frac{v_r}
    {k_o^2k_pk_v
    [1-\Delta\omega_o/(k_ok_v)]^2},
    \label{eq:nonlinear_inertia}
    \\
    D_o(\Delta\omega_o)
     & =
    \frac{v_r}
    {k_ok_pk_v
    [1-\Delta\omega_o/(k_ok_v)]}.
    \label{eq:nonlinear_damping}
\end{align}
Both coefficients are positive in this region. Equation~\eqref{eq:nonlinear_swing}, together with $\dot\delta=\Delta\omega_o-\Delta\omega$, is an exact swing-equation representation of the OBS dynamics. The equivalent inertia $M_o$ and $p$--$\omega$ coefficient $D_o$ emerge from the coupled voltage-magnitude and frequency dynamics and vary with the instantaneous frequency deviation. Equivalently, the estimated active power can be written as
\begin{equation}
    \hat p
    =
    p_r
    -D_o(\Delta\omega_o)\Delta\omega_o
    -M_o(\Delta\omega_o)\dot{\Delta\omega}_o.
    \label{eq:power_decomposition}
\end{equation}
The second and third terms represent the frequency-dependent and inertial power contributions, respectively. At synchronized steady state, the inertial term vanishes and $D_o$ determines the nonlinear droop relation for the estimated active power, as derived below. In particular, a frequency-decreasing transient satisfies $\dot{\Delta\omega}_o<0$ and therefore produces a positive inertial power contribution. At $\Delta\omega_o=0$, the equivalent coefficients reduce to
\begin{equation}
    M_{o,0}=\frac{v_r}{k_o^2k_pk_v},
    \qquad
    D_{o,0}=\frac{v_r}{k_ok_pk_v}.
    \label{eq:local_inertia_damping}
\end{equation}
Because the dot denotes differentiation with respect to the normalized time
$\tau=\omega_b t$, $M_o$ in \eqref{eq:nonlinear_swing} is a normalized-time
inertia coefficient. In physical time, the coefficient multiplying
$\mathrm d\Delta\omega_o/\mathrm dt$ is $M_o/\omega_b$. Consequently, the
local frequency-response time constant is
$T_o=M_{o,0}/(\omega_bD_{o,0})=1/(\omega_bk_o)$.
\par\endgroup

\begingroup
\subsubsection{Intrinsic Droop Characteristics}
\color{black}
Use the superscript $\star$ to denote steady-state values. Consider a
synchronized equilibrium under a constant grid frequency, where
$\dot{\hat V}=\dot\delta=0$ and hence
$\Delta\omega_o=\Delta\omega$. Equation~\eqref{eq:obs_internal_frequency}
determines the equilibrium value of $\hat V$, while \eqref{dv}
gives the corresponding estimated active power:
\begin{align}
    \hat V^\star
     & =
    \frac{v_r}{1-\Delta\omega/(k_ok_v)},
    \label{eq:steady_voltage_frequency}
    \\
    \hat p^\star
     & =
    p_r
    -D_o(\Delta\omega)\Delta\omega.
    \label{eq:steady_power_frequency}
\end{align}
At equilibrium, the first equation of \eqref{symp} directly gives
\begin{equation}
    \Delta v^\star
    =
    j\frac{\Delta\omega}{k_o}\hat v^\star.
    \label{eq:steady_delta_v}
\end{equation}
Let the actual voltage magnitude be denoted by $V=|u_c|$. Using the
grid-synchronous form of \eqref{eq:capacitor_translation} and
\eqref{eq:steady_delta_v}, its synchronized steady-state value is
\begin{equation}
    V^\star
    =
    \left|
    \left(1+j\frac{\Delta\omega}{k_o}\right)\hat v^\star
    -j\Delta\omega l_fi^\star
    \right|.
    \label{eq:actual_voltage_frequency}
\end{equation}
\par

Define the actual active power as
$p=\Re(u_c\overline{i})$. Since
$\Re(-j\Delta\omega l_f|i|^2)=0$, its steady-state relation is
\begin{equation}
    p^\star
    =
    p_r
    -
    \left[
        D_o(\Delta\omega)
        +
        \frac{\hat q^\star}{k_o}
        \right]
    \Delta\omega.
    \label{eq:actual_power_frequency}
\end{equation}
Thus, $V$ exhibits an inherent $V$--$\omega$
characteristic, as described by \eqref{eq:actual_voltage_frequency}, while
the active power exhibits an inherent $p$--$\omega$ droop characteristic,
as described by \eqref{eq:actual_power_frequency}. Linearizing the synchronized relations about $\Delta\omega=0$ gives
\begin{align}
    \hat V^\star
     & =
    v_r+K_{\hat V}\Delta\omega
    +\mathcal O(\Delta\omega^2),
    \\
    p^\star
     & =
    p_r-
    \left(D_{o,0}+\frac{\hat q_0}{k_o}\right)\Delta\omega
    +\mathcal O(\Delta\omega^2),
    \\
    V^\star
     & =
    v_r+K_V\Delta\omega
    +\mathcal O(\Delta\omega^2),
    \label{eq:local_steady_state_expansion}
\end{align}
where the estimated reactive power at $\Delta\omega=0$ is denoted by $\hat q_0$.
The local $\hat V$--$\omega$ and $V$--$\omega$ sensitivities are, respectively,
\begin{equation}
    K_{\hat V}=\frac{v_r}{k_ok_v},
    \qquad
    K_V=K_{\hat V}-\frac{l_f\hat q_0}{v_r}.
    \label{eq:capacitor_voltage_sensitivity}
\end{equation}
Accordingly, $D_{o,0}+\hat q_0/k_o$ is the local
$p$--$\omega$ sensitivity. Unlike $K_{\hat V}$, which is fixed directly by
the controller gains, both the $p$--$\omega$ sensitivity and
$K_V$ contain an operating-point- and network-dependent reactive-power term.
Equation~\eqref{eq:steady_delta_v} also shows that the steady-state feedback term $\Delta v^\star$ satisfies
\begin{equation}
    \frac{|\Delta v^\star|}{|\hat v^\star|}
    =
    \frac{|\Delta\omega|}{k_o}
    =
    \frac{M_{o,0}}{D_{o,0}}|\Delta\omega|.
    \label{eq:steady_correction_ratio}
\end{equation}
Thus, a larger ratio $M_{o,0}/D_{o,0}$ requires a larger $|\Delta v^\star|/|\hat v^\star|$ at a given frequency deviation.
\par

At $\Delta\omega=0$, \eqref{eq:steady_voltage_frequency} and
\eqref{eq:actual_voltage_frequency} give $\hat V^\star=V^\star=v_r$,
with no reactive-power-dependent voltage adjustment. The $\hat q_0$ term
in \eqref{eq:capacitor_voltage_sensitivity} instead describes how reactive
power affects the $V$--$\omega$ sensitivity. Thus, the original OBS
formulation in \cite{obGFM} does not provide a $q$--$V$ droop law, which
is needed for reactive-power sharing during parallel operation.
Incorporating this characteristic by extending the voltage correction
is left for future work.
\par\endgroup

\begingroup
\subsection{OBS Parameterization}
\label{subsec:obs_parameterization}
\label{subsec:obs_parameter_design}

\color{black}
The preceding relations provide a direct parameterization of the OBS.
Let the desired nominal $p$--$\omega$ coefficient be denoted by $D_p:=D_{o,0}$
and the desired $\hat V$--$\omega$ sensitivity by $K_{\hat V}$.
Equations~\eqref{eq:local_inertia_damping} and
\eqref{eq:capacitor_voltage_sensitivity} give
\begin{equation}
    M_{o,0}=\frac{D_p}{k_o},\qquad
    k_p=\frac{K_{\hat V}}{D_p},\qquad
    k_v=\frac{v_r}{k_oK_{\hat V}}.
    \label{eq:obs_oriented_design}
\end{equation}

Accordingly, $D_p$ and $K_{\hat V}$ specify the nominal $p$--$\omega$ coefficient
and $\hat V$--$\omega$ sensitivity, respectively, while
$k_o$ determines the equivalent inertia through
$M_{o,0}=D_p/k_o$. A smaller $k_o$ therefore provides larger inertia,
but also strengthens the nonlinear $V$--$\omega$ and $p$--$\omega$
characteristics in \eqref{eq:actual_voltage_frequency} and
\eqref{eq:actual_power_frequency}, respectively.
\par\endgroup

\begingroup\color{black}
\subsection{Relation to Established GFM Synchronization Laws}
\label{subsec:relation_vsg_droop}

The closest dynamical correspondence to OBS is VSG control. A conventional
VSG explicitly implements the  swing equation \cite{Zhong2011Synchronverter,paras}
\begin{equation}
    M_v\dot{\Delta\omega}_v
    =
    p_r-p_v-D_v\Delta\omega_v,
    \qquad
    \dot\delta_v=\Delta\omega_v-\Delta\omega,
    \label{eq:vsg_swing}
\end{equation}
which has the same swing-equation structure as
\eqref{eq:nonlinear_swing}. The difference is that $M_v$ and $D_v$ are
prescribed constants, whereas $M_o(\Delta\omega_o)$ and
$D_o(\Delta\omega_o)$ emerge from the OBS dynamics and are state dependent.

A conventional active-power Droop controller instead uses
\cite{Chandorkar1993Parallel,SimpsonPorco2013DroopSync}
\begin{equation}
    \Delta\omega_d=\frac{p_r-p_d}{D_d},
    \qquad
    \dot\delta_d=\Delta\omega_d-\Delta\omega.
    \label{eq:droop_law}
\end{equation}
where $D_d>0$ is the Droop $p$--$\omega$ coefficient.
At synchronized steady state, VSG and Droop control respectively give
\begin{equation}
    p_v^\star=p_r-D_v\Delta\omega,
    \qquad
    p_d^\star=p_r-D_d\Delta\omega.
    \label{eq:vsg_droop_steady}
\end{equation}

OBS similarly exhibits the nonlinear $p$--$\omega$ relation in
\eqref{eq:actual_power_frequency}. Thus, OBS, VSG, and Droop all exhibit
synchronized $p$--$\omega$ characteristics, although the OBS relation is
emergent, nonlinear, and operating-point dependent.

At a broader structural level, OBS is also related to dVOC/complex-droop
control \cite{Colombino2019GlobalSync,Seo2019DispatchableVOC,He2024ComplexDroop}.
Using the same notation as in \cite{He2024ComplexDroop}, the complex-droop
dynamics in complex-angle coordinates can be written as
\begin{equation}
    \frac{\mathrm d\underline{\vartheta}}{\mathrm dt}
    =
    \underline{\varpi}_0
    +
    \eta e^{j\varphi}
    \left(
    \overline{\underline{\varsigma}}^\star
    -
    \overline{\underline{\varsigma}}
    \right)
    +
    \eta\alpha
    \frac{v^{\star2}-v^2}{v^{\star2}}.
    \label{eq:complex_droop}
\end{equation}
The main structural relations are summarized in
Table~\ref{tab:obs_gfm_comparison}.
\par\endgroup

\begin{table*}[!t]
    \centering
    \caption{\color{black}Structural Relations Among Representative GFM Synchronization Laws}
    \label{tab:obs_gfm_comparison}
    \color{black}
    \scriptsize
    \setlength{\tabcolsep}{3.2pt}
    \renewcommand{\arraystretch}{1.16}

    \begin{tabular}{
        >{\raggedright\arraybackslash}p{0.16\textwidth}
        >{\raggedright\arraybackslash}p{0.21\textwidth}
        >{\raggedright\arraybackslash}p{0.16\textwidth}
        >{\raggedright\arraybackslash}p{0.16\textwidth}
        >{\raggedright\arraybackslash}p{0.24\textwidth}}
        \toprule
        Property
         & OBS \cite{obGFM}
         & VSG \cite{Zhong2011Synchronverter,paras}
         & Droop \cite{Chandorkar1993Parallel,SimpsonPorco2013DroopSync}
         & \mbox{dVOC / complex droop \cite{Colombino2019GlobalSync,Seo2019DispatchableVOC,He2024ComplexDroop}}
        \\
        \midrule

        Representative model
         &
        OBS dynamics
        \eqref{dv}--\eqref{dtheta};
        equivalent swing equation
        \eqref{eq:nonlinear_swing}
         &
        VSG model
        \eqref{eq:vsg_swing}
         &
        Droop law
        \eqref{eq:droop_law}
         &
        Complex-angle model
        \eqref{eq:complex_droop}
        \\

        Swing-equation coefficients
         &
        Emergent and state dependent:
        $M_o(\Delta\omega_o)$,
        $D_o(\Delta\omega_o)$
         &
        Prescribed and constant:
        $M_v$, $D_v$
         &
        No inertial state
         &
        No inertial state
        \\

        $p$--$\omega$ droop relation
         &
        Nonlinear; local coefficient
        $D_{o,0}+\hat q_0/k_o$
         &
        Linear; coefficient $D_v$
         &
        Linear; coefficient $D_d$
         &
        Embedded in complex-power feedback
        \\

        $q$--$V$ droop relation
         &
        No explicit droop law in the original OBS
         &
        Can be added separately
         &
        Can be added separately
         &
        Embedded in complex-power feedback
        \\

        $V$--$\omega$ relation
         &
        Inherent, as given by \eqref{eq:steady_voltage_frequency}
         &
        None
         &
        None
         &
        None
        \\

        Direct phase--magnitude coupling
        in synchronization law
         &
        Yes; $\hat V$ directly generates frequency
         &
        None
         &
        None
         &
        Yes; through complex-frequency dynamics
        \\

        Primary correspondence with OBS
         &
        ---
         &
        Swing dynamics and $p$--$\omega$
         &
        $p$--$\omega$
         &
        Explicit nonlinear phase--magnitude dynamics
        \\

        \bottomrule
    \end{tabular}
\end{table*}

\begingroup\color{black}
In Section~\ref{sec:validation}, the VSG and Droop parameters are set as follows:
\begin{equation}
    D_v=D_p,\quad
    M_v=M_{o,0},\quad
    D_d=D_p.
    \label{eq:obs_to_benchmarks}
\end{equation}
No explicit $q$--$V$ droop is introduced in either reference controller.
Thus, the VSG and Droop references are matched to the nominal OBS design
coefficient $D_p$, while OBS and VSG additionally share the same local inertia
and nominal $p$--$\omega$ coefficient. dVOC/complex droop is included here to clarify the broader
structural relation with OBS, but is not used as a quantitative reference in
Section~\ref{sec:validation} because the matching adopted above is defined in terms of the
$p$--$\omega$ coefficient and swing-equation inertia.
\par\endgroup

\section{Nonlinear Stability Analysis}
\label{sec:nonlinear_stability}
\label{subsec:nonlinear_stability}
This section establishes the non-linear convergence condition of OBS control based on the geometric and qualitative method.

\color{black}
Consider a constant grid-frequency deviation $\Delta\omega$, with the
equivalent parameters $U$, $Z$, and $\phi$ defined above. Under the
quasi-steady network approximation and $\Delta v_{\mathrm{tc}}\simeq0$,
the reduced OBS dynamics \eqref{dv}--\eqref{dtheta} can be rewritten as
\begin{equation}
    \dot{\hat V}=\frac{N(\hat V,\delta)}{G(\hat V)},
    \qquad
    \dot\delta=\kappa\frac{\hat V-\hat V_\omega}{\hat V},
    \label{aeq}
\end{equation}
where $\kappa=k_ok_v-\Delta\omega$, $\hat V_\omega=k_ok_vv_r/\kappa$, $N(\hat V,\delta)=a\hat V^2+b(\delta)\hat V+c$, and $G(\hat V)=Z+k_p\hat V\cos\phi$, with
\begin{itemize}
    \item
          $a=-k_ok_p(\cos\phi+k_v\sin\phi)$,
    \item
          $b(\delta)=k_o[k_pU\cos(\delta+\phi)+k_pk_vv_r\sin\phi-Z]$,
    \item
          $c=k_oZ(k_pp_r+v_r)$.
\end{itemize}

For practical inductive grids and controller design, we assume that $a<0$, $c>0$, $G(\hat V)>0$ for all $\hat V>0$, and $\kappa>0$. Hence, the vector field of \eqref{aeq} is analytic on $\mathcal E=(0,\infty)\times\mathbb R$.

Define $b_{\max}=k_o(k_pU+k_pk_vv_r\sin\phi-Z)$ and $b_{\min}=k_o(-k_pU+k_pk_vv_r\sin\phi-Z)$. Let
\[
    \hat V_p(b)=\frac{-b-\sqrt{b^2-4ac}}{2a}
\]
be the positive root of $a\hat V^2+b\hat V+c=0$. Since $a<0$ and $c>0$, the other root is negative, and $\hat V_p(b)$ is strictly increasing in $b$. Define $\hat V_-=\hat V_p(b_{\min})$, $\hat V_+=\hat V_p(b_{\max})$, and $\mathcal K=[\hat V_-,\hat V_+]\times\mathbb R$.
The equilibrium condition gives $\cos(\delta+\phi)=\sigma$, where
\begin{equation}
    \sigma
    =
    \frac{
        \hat V_\omega^2\cos\phi
        +(k_v\hat V_\omega\sin\phi+Z/k_p)(\hat V_\omega-v_r)
        -Zp_r
    }{
        U\hat V_\omega
    }.
    \label{eq:sigma_def}
\end{equation}

For $-1<\sigma<1$, let $\alpha=\arccos\sigma\in(0,\pi)$, $\delta_s=\alpha-\phi$, and $\delta_u=-\alpha-\phi$. Then $\delta_u<\delta_s<\delta_u+2\pi$.
\begin{lemma}[Equilibrium classification]
    \label{lem:eq_classification}
    Suppose that $-1<\sigma<1$. The system has two equilibria within a $2\pi$ interval, $(\hat V_\omega,\delta_s)$ and $(\hat V_\omega,\delta_u)$. The former is locally asymptotically stable, the latter is a saddle, and $\hat V_-<\hat V_\omega<\hat V_+$.
\end{lemma}
\begin{proof}
    Since $\dot\delta=0$ if and only if $\hat V=\hat V_\omega$, the equilibrium phases satisfy $\cos(\delta+\phi)=\sigma$, yielding $\delta_s$ and $\delta_u$.
    Let $\mathcal J$ denote the Jacobian matrix of \eqref{aeq}. At either equilibrium,
    \[
        \begin{aligned}
            \operatorname{tr}\mathcal J
             & =
            -\frac{\sqrt{b^2(\delta_{\mathrm{eq}})-4ac}}
            {G(\hat V_\omega)}
            <0,
            \\
            \det\mathcal J
             & =
            \frac{k_ok_pU\kappa\sin(\delta_{\mathrm{eq}}+\phi)}
            {G(\hat V_\omega)}.
        \end{aligned}
    \]
    Hence $\delta_s$ is locally asymptotically stable, whereas $\delta_u$ is a saddle.
    Moreover, $-1<\sigma<1$ gives $b_{\min}<b(\delta_s)<b_{\max}$. Since $N(\hat V_\omega,\delta_s)=0$, one has $\hat V_\omega=\hat V_p(b(\delta_s))$.
    The monotonicity of $\hat V_p$ therefore gives $\hat V_-<\hat V_\omega<\hat V_+$.
\end{proof}

\begin{lemma}[Voltage trapping]
    \label{lem:voltage_trap}
    The set $\mathcal K$ is positively invariant, and every trajectory with $\hat V(0)>0$ enters $\mathcal K$.
\end{lemma}
\begin{proof}
    Since $b_{\min}\leq b(\delta)\leq b_{\max}$,
    \[
        \frac{a\hat V^2+b_{\min}\hat V+c}{G(\hat V)}
        \leq \dot{\hat V} \leq
        \frac{a\hat V^2+b_{\max}\hat V+c}{G(\hat V)}.
    \]
    Hence $\dot{\hat V}<0$ for $\hat V>\hat V_+$ and $\dot{\hat V}>0$ for $0<\hat V<\hat V_-$. On $\hat V=\hat V_+$, $\dot{\hat V}\leq0$, with equality only at the isolated phases where $b(\delta)=b_{\max}$; similarly, on $\hat V=\hat V_-$, $\dot{\hat V}\geq0$, with equality only where $b(\delta)=b_{\min}$. Since $\hat V_-<\hat V_\omega<\hat V_+$, one has $\dot\delta(\hat V_\pm)\neq0$, so these tangencies cannot persist. Therefore $\mathcal K$ is positively invariant, and every trajectory with $\hat V(0)>0$ enters $\mathcal K$.
\end{proof}
Define the phase-excursion regions
\begin{equation}
    \begin{aligned}
        \mathcal R_R
         & =
        \{(\hat V,\delta):
        \hat V_\omega<\hat V<\hat V_+,\;
        \delta_s<\delta<\delta_u+2\pi\},
        \\
        \mathcal R_L
         & =
        \{(\hat V,\delta):
        \hat V_-<\hat V<\hat V_\omega,\;
        \delta_u<\delta<\delta_s\}.
    \end{aligned}
    \label{eq:RR_RL_def}
\end{equation}

Since $G(\hat V)>0$ and $\hat V_p(b(\delta))$ is the unique positive root of $N(\hat V,\delta)=0$, the sign of $\dot{\hat V}$ is determined by the relative position of $\hat V$ and $\hat V_p(b(\delta))$. Moreover, $\hat V_p(b(\delta_s))=\hat V_\omega$.

In $\mathcal R_R$, $b(\delta)<b(\delta_s)$ and, since $\hat V_p$ is strictly increasing, $\hat V_p(b(\delta))<\hat V_\omega<\hat V$. Hence $\dot{\hat V}<0$, while \eqref{aeq} gives $\dot\delta>0$. Moreover, since $\hat V>\hat V_\omega=\hat V_p(b(\delta_s))$, one has $N(\hat V,\delta_s)<0$. Together with $N(\hat V,\delta)-N(\hat V,\delta_s) =[b(\delta)-b(\delta_s)]\hat V<0$, this yields $N(\hat V,\delta)<N(\hat V,\delta_s)<0$.

Similarly, in $\mathcal R_L$, one obtains $\hat V<\hat V_\omega<\hat V_p(b(\delta))$, $\dot{\hat V}>0$, $\dot\delta<0$, and $N(\hat V,\delta)>N(\hat V,\delta_s)>0$.
Thus, $|N(\hat V,\delta)|>|N(\hat V,\delta_s)|$ throughout $\mathcal R_R\cup\mathcal R_L$. The corresponding phase-plane geometry and sign distribution are illustrated in \textcolor{black}{Fig.~\ref{fig:RR_RL_regions}}.
\begin{figure}[!h]
    \centering
    \includegraphics[width=0.45\textwidth]{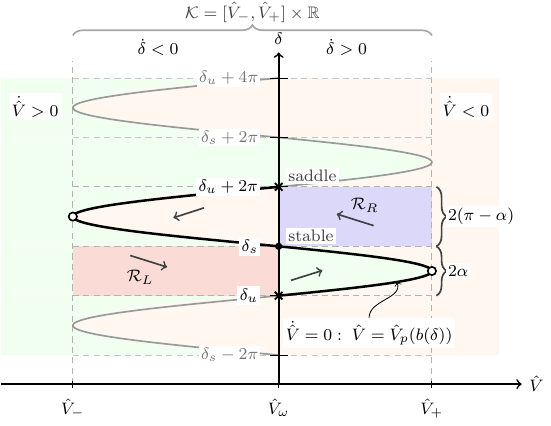}
    \caption{\color{black}Phase-plane geometry and sign distribution for $-1<\sigma<1$.}
    \label{fig:RR_RL_regions}
\end{figure}
For the phase-excursion estimate, $N(\hat V_\omega,\delta_s)=0$ gives $N(\hat V,\delta_s) =(\hat V-\hat V_\omega)(a\hat V-c/\hat V_\omega)$. Since $a<0$ and $c>0$, $c/\hat V_\omega-a\hat V>0$ for $\hat V>0$. Define
\begin{equation}
    \Gamma(\hat V)
    =
    \frac{\kappa G(\hat V)}
    {\hat V(c/\hat V_\omega-a\hat V)},
    \qquad \hat V>0.
    \label{eq:Gamma_def}
\end{equation}
Using $|N(\hat V,\delta)|>|N(\hat V,\delta_s)|$, the trajectory slope in $\mathcal R_R\cup\mathcal R_L$ satisfies
\begin{equation}
    \left|
    \frac{\mathrm d\delta}{\mathrm d\hat V}
    \right|
    =
    \frac{\kappa G(\hat V)|\hat V-\hat V_\omega|}
    {\hat V|N(\hat V,\delta)|}
    <
    \Gamma(\hat V).
    \label{eq:trajectory_slope_bound}
\end{equation}
Accordingly, define the phase-excursion bounds
\begin{equation}
    I_R:=\int_{\hat V_\omega}^{\hat V_+}\Gamma(\hat V)\,\mathrm d\hat V,
    \qquad
    I_L:=\int_{\hat V_-}^{\hat V_\omega}\Gamma(\hat V)\,\mathrm d\hat V.
    \label{eq:phase_excursion_bounds}
\end{equation}
\begin{lemma}[Phase boundedness]
    \label{lem:phase_boundedness}
    If $I_R<2(\pi-\alpha)$ and $I_L<2\alpha$, then every trajectory in $\mathcal K$ has bounded phase.
\end{lemma}
\begin{proof}
    By the $2\pi$-periodicity of the vector field, it is sufficient to consider crossings of $\delta=\delta_s$. Consider an upward crossing of $\delta=\delta_s$, which requires $\hat V_0\in(\hat V_\omega,\hat V_+]$. The trajectory then enters $\mathcal R_R$. By \eqref{eq:trajectory_slope_bound},
    \[
        \delta-\delta_s
        <
        \int_{\hat V_\omega}^{\hat V_0}\Gamma(\hat V)\,\mathrm d\hat V
        \leq I_R.
    \]
    Since $I_R<\delta_u+2\pi-\delta_s=2(\pi-\alpha)$, the trajectory reaches $\hat V=\hat V_\omega$ before $\delta=\delta_u+2\pi$, after which $\dot\delta$ reverses sign.

    Similarly, a downward crossing of $\delta_s$ enters $\mathcal R_L$. Since $I_L<\delta_s-\delta_u=2\alpha$, the trajectory reaches $\hat V=\hat V_\omega$ before $\delta=\delta_u$, after which its phase direction reverses.
    By periodicity, the same argument applies to every equivalent phase interval. Hence a trajectory cannot pass successively through adjacent $2\pi$ phase intervals, and its phase remains bounded.
\end{proof}
Following \cite{perko2013differential}, an analytic planar system is relatively
prime if its vector-field components have no nontrivial common analytic
factor. System~\eqref{aeq} is relatively prime and analytic on
$\mathcal E=(0,\infty)\times\mathbb R$.
\AddToHookNext{cmd/thmhead/before}{\color{black}}
\begin{lemma}[Generalized Poincar\'e--Bendixson theorem]
    \label{lem:PB}
    \textcolor{black}{Consider a two-dimensional relatively prime analytic system on an open set $\Omega\subset\mathbb R^2$. If a positive trajectory is contained in a compact subset of $\Omega$, then its $\omega$-limit set is either an equilibrium, a periodic orbit, or a graphic \cite{perko2013differential}.}
\end{lemma}
A periodic orbit and a graphic are illustrated in
\textcolor{black}{Fig.~\ref{fig:cycle_types}}.

\begin{figure}[!h]
    \centering
    \includegraphics[width=0.38\textwidth]{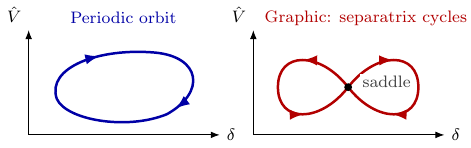}
    \caption{\color{black}Illustration of a periodic orbit and a graphic.}
    \label{fig:cycle_types}
\end{figure}

Physically, a periodic orbit corresponds to sustained oscillations of the
observer-voltage magnitude and synchronization angle. Following
\cite{perko2013differential}, a graphic consists of equilibrium points and
connecting trajectories that form one or more closed loops with compatible
flow directions. Such a limiting set would also prevent convergence to a single
synchronized equilibrium. Lemma~\ref{lem:dulac} below excludes both
possibilities for the reduced OBS model under the stated assumptions.
\AddToHookNext{cmd/thmhead/before}{\color{black}}
\begin{lemma}[Dulac's criterion]
    \label{lem:dulac_criteria}
    \textcolor{black}{Consider a planar system $\dot x=f(x)$ on a simply connected region $\Omega\subset\mathbb R^2$. If there exists a function $\mu\in C^1(\Omega)$ such that $\nabla\cdot(\mu f)$ is not identically zero and does not change sign in $\Omega$, then the system has no closed orbit in $\Omega$.
    Under the same sign condition, separatrix cycles or graphics are also excluded \cite{perko2013differential}.}
\end{lemma}
\begin{lemma}[Exclusion of Periodic Orbits and Graphics]
    \label{lem:dulac}
    System \eqref{aeq} admits neither a periodic orbit nor a graphic.
\end{lemma}
\begin{proof}
    Choose the Dulac function $\mu(\hat V,\delta)=G(\hat V)/\hat V$. Then $\mu\dot{\hat V}=a\hat V+b(\delta)+c/\hat V$, while $\mu\dot\delta$ is independent of $\delta$. Hence
    \[
        \frac{\partial(\mu\dot{\hat V})}{\partial\hat V}
        +
        \frac{\partial(\mu\dot\delta)}{\partial\delta}
        =
        a-\frac{c}{\hat V^2}
        <0.
    \]
    Since the domain $\mathcal E=(0,\infty)\times\mathbb R$ is simply connected, Lemma~\ref{lem:dulac_criteria} implies that system \eqref{aeq} has neither a periodic orbit nor a graphic.
\end{proof}
\begin{theorem}[Almost-global convergence of the reduced model]
    \label{thm:almost_global_stability}
    Suppose that $-1<\sigma<1$, $I_R<2(\pi-\alpha)$, and $I_L<2\alpha$.
    Then, for the reduced system \eqref{aeq}, almost every trajectory with $\hat V(0)>0$ converges to one member of the stable equilibrium family $\{(\hat V_\omega,\delta_s+2n\pi): n\in\mathbb Z\}$, except for initial conditions on the stable manifolds of the saddle family $\{(\hat V_\omega,\delta_u+2n\pi): n\in\mathbb Z\}$, which form a set of zero measure \cite{perko2013differential}.
\end{theorem}
\begin{proof}
    By Lemmas~\ref{lem:voltage_trap} and~\ref{lem:phase_boundedness}, every trajectory with $\hat V(0)>0$ is eventually contained in a compact subset of $\mathcal E$. Lemma~\ref{lem:PB} then implies that its $\omega$-limit set is an equilibrium, a periodic orbit, or a graphic. The latter two are excluded by Lemma~\ref{lem:dulac}; hence every trajectory converges to an equilibrium.
    By $2\pi$-periodicity, the equilibria form the stable family $(\hat V_\omega,\delta_s+2n\pi)$ and the saddle family $(\hat V_\omega,\delta_u+2n\pi)$, $n\in\mathbb Z$. Since the stable manifolds of the saddle family have zero measure, almost every trajectory converges to a member of the stable equilibrium family.
\end{proof}

The overall proof structure is summarized in \textcolor{black}{Fig.~\ref{fig:proof_logic}}. Voltage trapping and phase boundedness ensure compactness of the trajectories, while Dulac's criterion excludes periodic orbits and graphics. Consequently, except for the stable manifolds of the saddle family, all trajectories converge to the stable equilibrium family.
\begin{figure}[!h]
    \centering
    \includegraphics[width=0.48\textwidth]{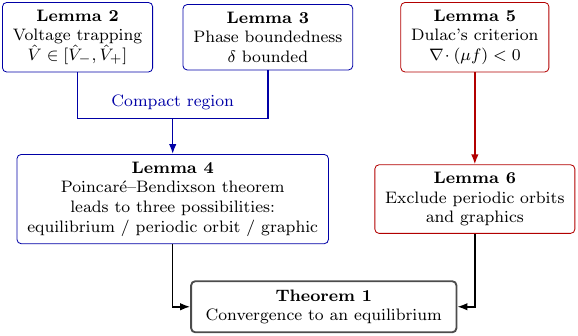}
    \caption{\color{black}Overall proof structure for the $\omega$-limit-set analysis.}
    \label{fig:proof_logic}
\end{figure}
The above conditions provide explicit sufficient parameter conditions for almost-global convergence to the stable equilibrium family of the reduced synchronization model.
\color{black} 

\section{Numerical and EMT Validation}
\label{sec:validation}

This section evaluates the OBS design implications established in
Section~\ref{subsec:obs_parameterization} and the nonlinear stability results
in Section~\ref{sec:nonlinear_stability} using the reduced-order system
in \eqref{aeq}, followed by EMT comparisons with the VSG and Droop
controllers. The numerical studies are organized around the OBS design
parameters $(D_p,k_o,K_{\hat V})$, while the VSG and Droop controllers are
parameterized according to \eqref{eq:obs_to_benchmarks} for comparison.

The common system parameters are summarized in
Table~\ref{tab:system_params} \cite{paras}. The short-circuit ratio (SCR) is defined as the reciprocal
of the per-unit grid-impedance magnitude. The grid resistance and
inductance are varied according to the SCR while
maintaining a constant $X/R$ ratio of 7.
Unless otherwise specified, the operating point is selected as
$p_r=0.7$ p.u. and $v_r=1.0$ p.u. The choice of $D_p$ is based on the maximum allowable frequency
deviation of $0.2$ Hz specified in \cite{GBT15945}.
Under the nominal $p$--$\omega$ relation at a nominal frequency of
50 Hz, this corresponds to a 1-p.u. active-power deviation. The controller
parameters are summarized in
Table~\ref{tab:obs_params}.

\begin{table}[!t]
    \centering
    \color{black}
    \caption{Common Simulation Parameters}
    \label{tab:system_params}
    \begin{tabular}{ll}
        \toprule
        Parameter                    & Value                                                \\
        \midrule
        Base voltage / power / freq. & 690 V / 5 MW / $100\pi$ rad/s                        \\
        Filter $r_f/l_f/c_f$         & $0.01~\Omega$ / $95~\mu\mathrm{H}$ / $1~\mathrm{mF}$ \\
        Grid $X/R$ ratio             & 7                                                    \\
        DC-link voltage              & 1500 V                                               \\
        \bottomrule
    \end{tabular}
\end{table}

\begin{table}[!t]
    \centering
    \color{black}
    \caption{OBS, VSG, and Droop Parameters}
    \label{tab:obs_params}
    \setlength{\tabcolsep}{5pt}
    \renewcommand{\arraystretch}{1.05}
    \begin{tabular}{@{}llcc@{}}
        \toprule
        Controller & Parameter                  & Setting 1 & Setting 2 \\
        \midrule
        OBS        & $D_p$                      & 250       & 250       \\
                   & $K_{\hat V}$               & 5         & 5         \\
                   & $k_o$                      & 0.1       & 0.01      \\
                   & $k_p=K_{\hat V}/D_p$       & 0.02      & 0.02      \\
                   & $k_v=v_r/(k_oK_{\hat V})$  & 2         & 20        \\
                   & $M_{o,0}=D_p/k_o$          & 2500      & 25000     \\
                   & $T_o=1/(\omega_bk_o)$ (ms) & 31.8      & 318.3     \\
        \midrule
        VSG        & $M_v=M_{o,0}$              & 2500      & 25000     \\
                   & $D_v=D_p$                  & 250       & 250       \\
        Droop      & $D_d=D_p$                  & 250       & 250       \\
        \bottomrule
    \end{tabular}
\end{table}

\color{black}
For the OBS implementation, the transient-current feedback parameters are
selected as $r_a=0.5$ p.u. and $k_a=10$ p.u. The VSG and Droop controllers use
the same nested voltage--current control structure. Following the classical
bandwidth-separation principle in
\cite{Yazdani2010VSC,Liu2025SwitchingFrequency,DArco2015CascadedTuning},
the current-loop bandwidth and voltage-loop natural frequency are selected as
$\omega_i=10$ p.u. and $\omega_v=6$ p.u., respectively, with the voltage-loop pole pair specified by
$\xi=0.707$. The corresponding PI gains are calculated as
$k_{pc}=\omega_i l_f$, $k_{ic}=\omega_i r_f$,
$k_{pv}=2\xi\omega_v c_f$, and $k_{iv}=\omega_v^2c_f$, resulting in
$k_{pc}=3.13$, $k_{ic}=1.05$, $k_{pv}=0.25452$, and $k_{iv}=1.08$. The resulting closed-loop voltage bandwidths, evaluated from the first
$-3$-dB crossing of the voltage response, are $30.730$, $22.168$, and
$7.045$ Hz for $\mathrm{SCR}=1$, $2$, and $10$, respectively.

\color{black}
\begingroup\color{black}
\subsection{Reduced-Order Numerical Analysis}
\label{subsec:reduced_numerical_validation}

\subsubsection{Accuracy of the Reduced-Order Model}
\label{subsubsec:reduced_accuracy}

Fig.~\ref{fig:emt_reduced_comparison} shows representative EMT and reduced-model
responses to an active-power-reference step in (a) and a deep grid-voltage
sag in (b), at $\mathrm{SCR}=1$ and $k_o=0.1$.
The disturbances are $p_r:0.7\rightarrow1.0$ p.u. and a grid-voltage
sag from 1.0 to 0.2 p.u. at $p_r=0.7$ p.u. Both disturbances start at
$t=0.3$ s, and the sag ends at $t=0.6$ s. The shaded interval marks
the voltage sag.
Both models are initialized at a consistent operating point. Table~\ref{tab:reduced_model_error} reports the maximum
absolute errors $e_{x,\max}=\max_t|x_{\mathrm{EMT}}(t)-x_{\mathrm{red}}(t)|$
over the 5-s simulations. The reduced model shows very close alignment
with the EMT results in both the observer states $(\hat V,\delta)$ and
the electrical quantities $p$ and $|u_c|$ under the power-reference step.
While $p$ and $|u_c|$ exhibit
appreciable transient discrepancies during the deep voltage sag and
restoration, the EMT and reduced-model results for $\hat V$ and $\delta$
remain quite close.
These results show that the reduced and EMT models exhibit consistent
long-term behavior in the tested scenarios, while short-time discrepancies
arise from the fast electrical dynamics omitted by the quasi-steady-state
approximation.
\par

\begin{table}[!t]
    \centering
    \caption{\color{black}Maximum EMT--Reduced Model Errors}
    \label{tab:reduced_model_error}
    \setlength{\tabcolsep}{3pt}
    \color{black}
    \begin{tabular}{@{}lcc@{}}
        \toprule
        Maximum error            & Power step          & Voltage sag         \\
        \midrule
        $e_{\hat V,\max}$ (p.u.) & $1.23\times10^{-5}$ & $2.79\times10^{-3}$ \\
        $e_{\delta,\max}$ (rad)  & $2.01\times10^{-4}$ & $3.53\times10^{-3}$ \\
        $e_{p,\max}$ (p.u.)      & $5.04\times10^{-3}$ & $6.96\times10^{-1}$ \\
        $e_{|u_c|,\max}$ (p.u.)  & $5.98\times10^{-3}$ & $2.50\times10^{-1}$ \\
        \bottomrule
    \end{tabular}
\end{table}

\begin{figure}[!t]
    \centering
    \includegraphics[width=\columnwidth]{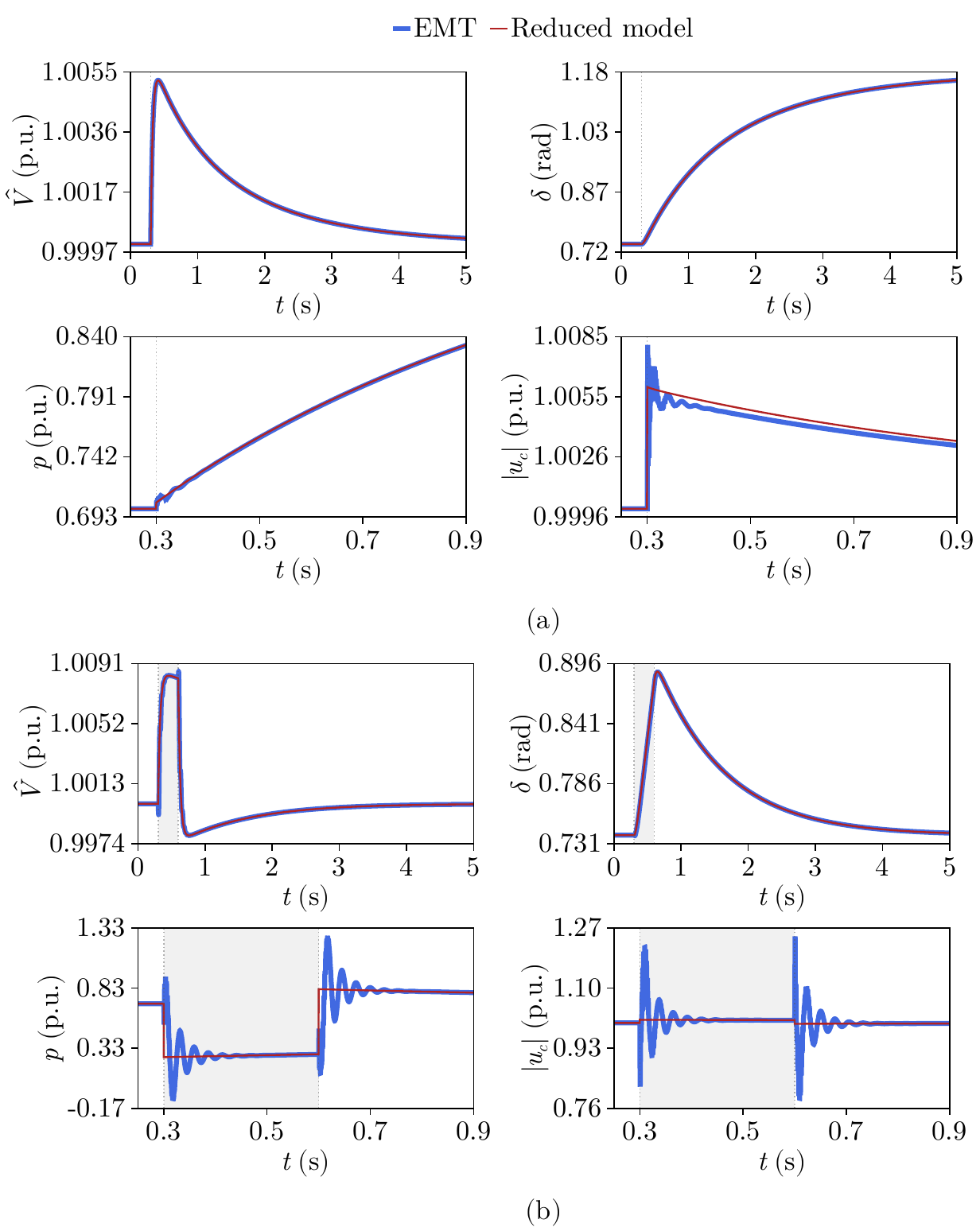}
    \caption{\color{black}EMT and reduced-model comparison: (a) power-reference step; (b) voltage sag.}
    \label{fig:emt_reduced_comparison}
\end{figure}

\par\endgroup

\begingroup\color{black}
\subsubsection{\textcolor{black}{Droop Characteristics and the Role of $k_o$}}
\label{subsubsec:steady_state_ko}

\textcolor{black}{Fig.~\ref{fig:steady_state_validation} compares the nonlinear
    $p$--$\omega$ and $V$--$\omega$ characteristics with their local
    approximations at $\mathrm{SCR}=5$, with shaded bands spanning the
    synchronized equilibria for $\mathrm{SCR}=1,2,5,10$.
    Reducing $k_o$ increases the nominal inertia at the cost of stronger
    nonlinear droop, particularly in the $V$--$\omega$ relation.
    For the two configurations in Table~\ref{tab:obs_params}, reducing $k_o$
    from 0.1 to 0.01 increases $M_{o,0}=D_p/k_o$ to ten times its original value, while
    $|\Delta v^\star|/|\hat v^\star|$ increases from $4\%$ to $40\%$ at
    $|\Delta f|=0.2$ Hz according to \eqref{eq:steady_correction_ratio}.}

\begin{figure}[!t]
    \centering
    \includegraphics[width=\columnwidth]{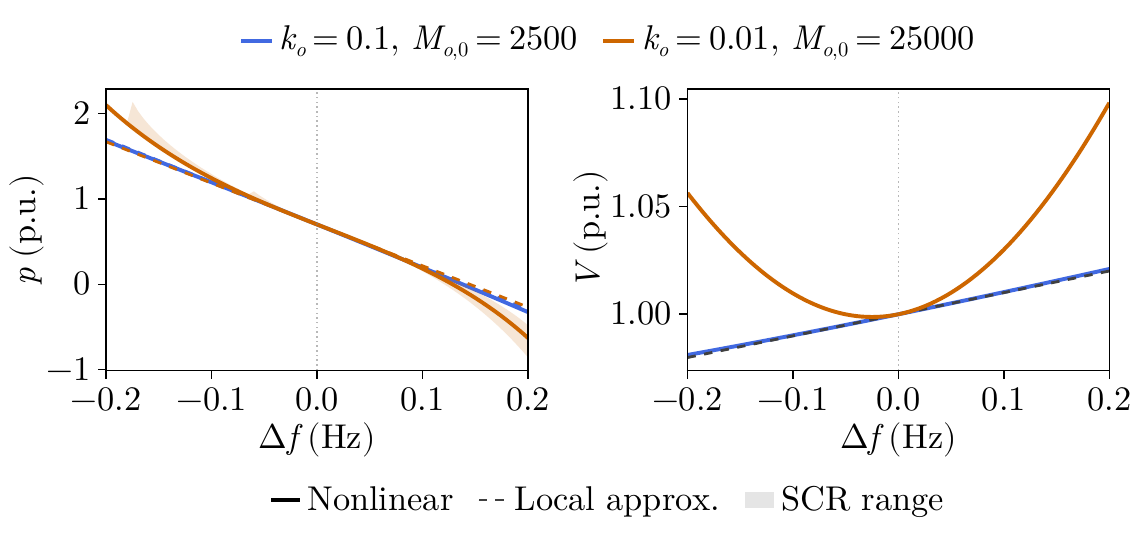}
    \caption{\color{black}Nonlinear droop characteristics and their local approximations.}
    \label{fig:steady_state_validation}
    %
\end{figure}

%
\par\endgroup

\color{black}
\subsubsection{Numerical Evaluation of the AGAS Certificate}
\label{subsubsec:agas_numerical}
The nonlinear stability results are examined through phase-plane
trajectories and convergence regions.

\paragraph{Phase-Plane Verification}

Fig.~\ref{fig:stability_phase_plane} shows
representative trajectories for $\mathrm{SCR}=2$ and 10,
$\Delta f=0$ Hz, $p_r=0.7$ p.u., and $k_o=0.1$. The analytical bounds $\hat V_-$ and $\hat V_+$, the equilibrium voltage
$\hat V_\omega$, and the stable and saddle equilibria are determined from the
analysis in Section~\ref{sec:nonlinear_stability}.

\begin{figure}[!t]
    \centering
    \includegraphics[width=\columnwidth]{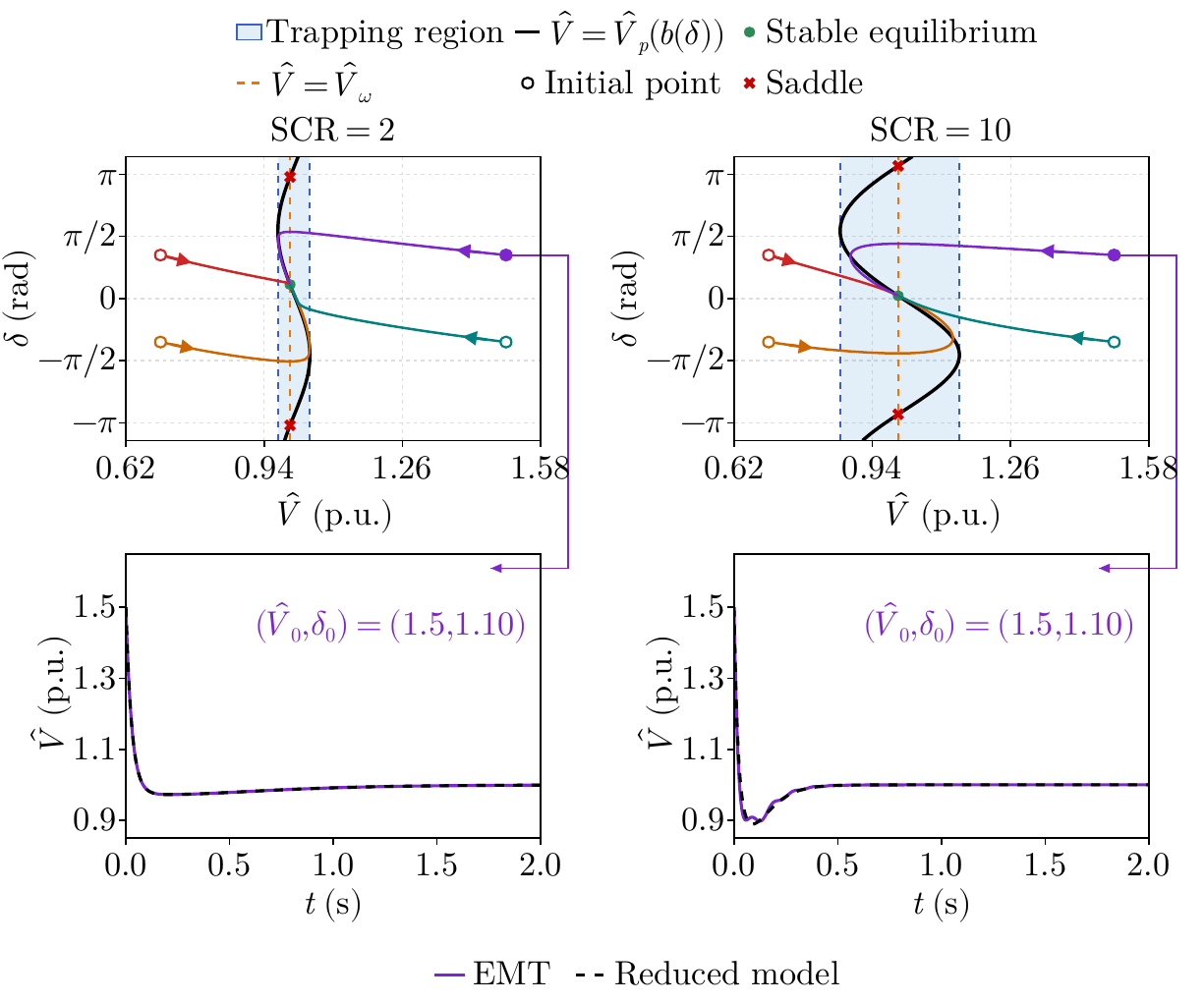}
    \caption{\color{black}Phase-plane trajectories with analytical trapping bounds and comparison of EMT and reduced-model $\hat V$ responses.}
    \label{fig:stability_phase_plane}
\end{figure}

\textcolor{black}{Both cases satisfy the AGAS conditions, with the trapping interval
$[\hat V_-,\hat V_+]$ widening from
$[0.9719,1.0459]$ p.u. at $\mathrm{SCR}=2$ to
$[0.8656,1.1416]$ p.u. at $\mathrm{SCR}=10$.
The curves $\hat V=\hat V_\omega$ and $\hat V=\hat V_p(b(\delta))$
intersect at the stable and saddle equilibria.
Trajectories initialized outside the trapping interval enter it and remain
inside. They then approach the $\dot{\hat V}=0$ nullcline
$\hat V=\hat V_p(b(\delta))$ and evolve near it toward the stable
equilibrium, consistent with the phase-plane behavior described in the
AGAS proof.
For the selected initial point $(\hat V,\delta)=(1.5,1.10)$, the bottom row
shows consistent convergence of $\hat V$ in the EMT and reduced models.}
\begin{figure}[!t]
    \centering
    \includegraphics[width=\columnwidth]{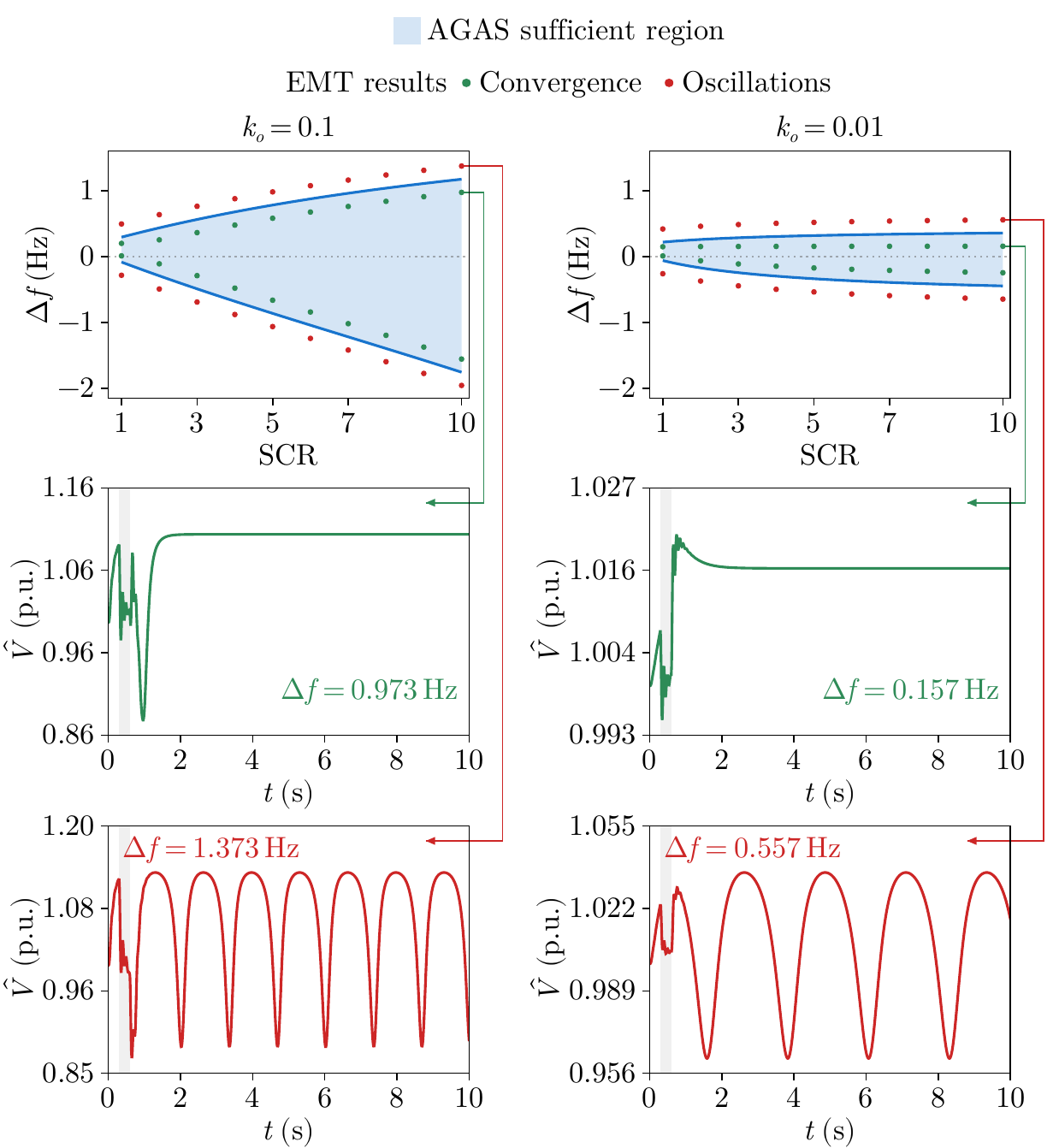}
    \caption{\color{black}AGAS sufficient regions and EMT responses to a grid-voltage sag.}
    \label{fig:parameter_stability_area}
\end{figure}

\paragraph{Certified Stability Region}

\color{black}
Fig.~\ref{fig:parameter_stability_area} shows the regions satisfying the
sufficient conditions of Theorem~\ref{thm:almost_global_stability} over
$\mathrm{SCR}=1$--$10$ for $k_o=0.1$ and 0.01, with $p_r=0.7$ p.u.
and $v_r=1.0$ p.u. Four EMT cases are selected at each SCR, one on either side
of each boundary. The EMT model is initialized at the references defined above.
A grid-voltage sag from 1.0 to 0.2 p.u.
is applied during $t=0.3$--$0.6$ s, as indicated by the shaded intervals.
After voltage restoration, the green cases inside the certified region
converge to synchronized equilibria, while the selected red cases
outside it exhibit oscillations. The EMT responses are consistent with
the theoretical analysis. We show two representative cases near the
upper boundary at $\mathrm{SCR}=10$ for each $k_o$ in the lower panels.
\par\color{black}

\subsection{EMT Reference Study with VSG and Droop}
\label{subsec:emt_validation}

\begin{figure}[!t]
    \centering
    \includegraphics[width=\columnwidth]{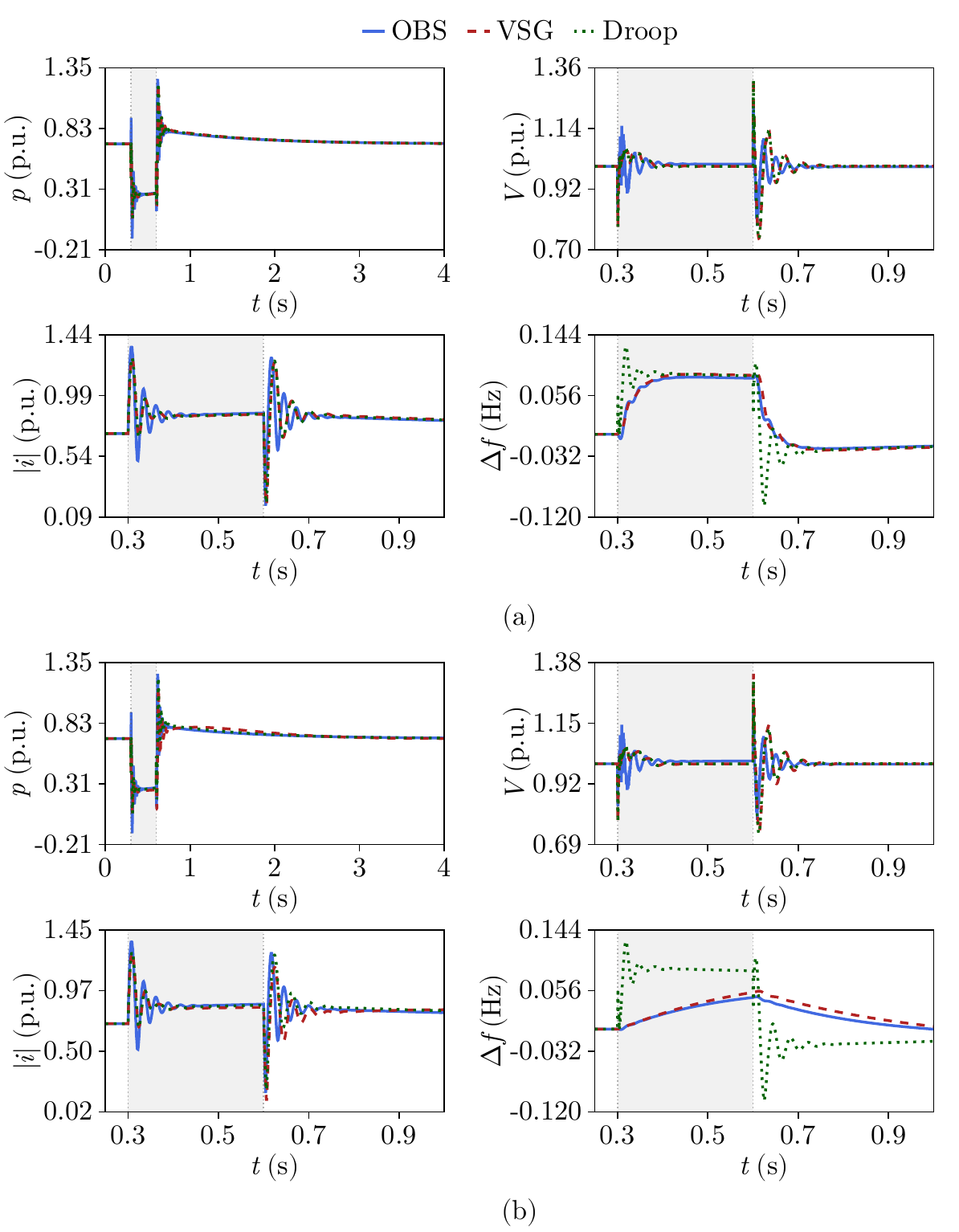}
    \caption{\color{black}EMT voltage-sag responses of OBS, VSG, and Droop for
        (a) $k_o=0.1$ and (b) $k_o=0.01$.}
    \label{fig:emt_combined}
\end{figure}

\textcolor{black}{All comparisons in this subsection use the controller parameters listed in Table~\ref{tab:obs_params} and are conducted at $\mathrm{SCR}=1$. All controllers are initialized at the same nominal operating point, with $p_r=0.7$ p.u. and $v_r=1.0$ p.u.}

\textcolor{black}{Fig.~\ref{fig:emt_combined} compares the responses to a grid-voltage sag
from 1.0 to 0.2 p.u. during $t=0.3$--$0.6$ s. Reducing $k_o$ from 0.1
to 0.01 decreases the peak $|\Delta f|$ from 0.085 to 0.047 Hz for OBS
and from 0.087 to 0.055 Hz for VSG. The peak $|\Delta f|$ is
0.128 Hz with Droop control. In contrast, the peak $|i|$ remains approximately 1.36 p.u.
for OBS and 1.27 p.u. for VSG and Droop. Across both gain settings,
the ranges of $V$ are approximately $[0.78,1.23]$ p.u. for OBS,
$[0.74,1.34]$ p.u. for VSG, and $[0.74,1.31]$ p.u. for Droop.
All three controllers recover toward the prefault operating point within
the 4-s simulation.}

\textcolor{black}{These results demonstrate that OBS provides an inertial frequency response comparable to that of VSG. Increasing the equivalent inertia reduces the peak frequency deviation for both controllers. The slightly smaller frequency deviation is related to the nonlinear characteristics of OBS, as in \eqref{eq:nonlinear_inertia} and \eqref{eq:nonlinear_damping}. Compared with VSG and Droop, OBS exhibits a smaller voltage excursion during the tested grid-voltage sag, accompanied by a slightly higher peak current. Thus, the OBS parameterization allows the frequency response to be adjusted through its intrinsic inertia, while the voltage and current responses should also be considered when evaluating its disturbance performance.}

\begin{figure}[!t]
    \centering
    \includegraphics[width=\columnwidth,pagebox=cropbox]{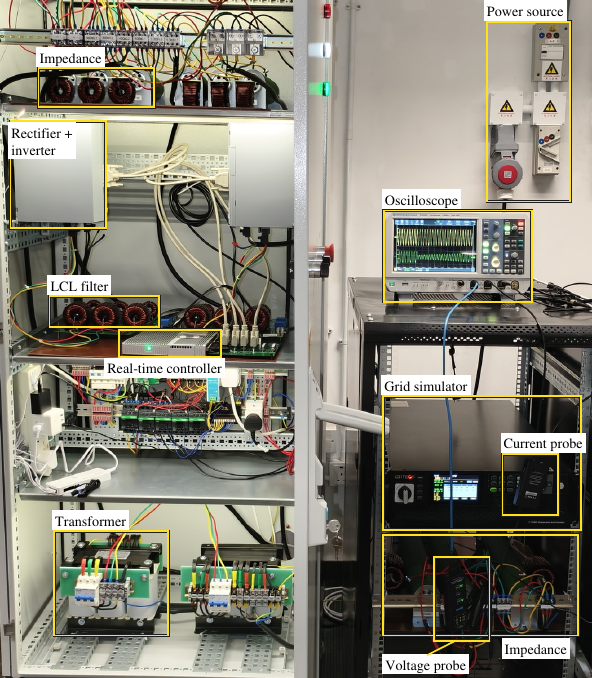}
    \caption{\textcolor{black}{Experimental setup of the \SI{5.5}{\kilo\watt} three-phase inverter.}}
    \label{expe_setup}
\end{figure}

\section{Experimental Validation}
\label{sec:experimental_validation}

Experiments are conducted on the \SI{5.5}{\kilo\watt} three-phase inverter
platform shown in Fig.~\ref{expe_setup}, while Fig.~\ref{expe_circuit}
shows the experimental circuit and measurement connections.
The OBS controller is implemented on
an InsRealm RTScale real-time control platform, which generates the inverter
PWM signals. The inverter is connected to the grid simulator through an LCL
filter and a grid impedance. The nominal grid phase-voltage RMS is
\SI{150}{\volt}, and the nominal dc-link voltage is \SI{580}{\volt}. The LCL filter
consists of a \SI{6}{\milli\henry} inverter-side inductor, a
\SI{10}{\micro\farad} capacitor with a \SI{1}{\ohm} series resistor,
and a \SI{300}{\micro\henry} grid-side inductor. The grid impedance is
\SI{16.3}{\milli\henry} and \SI{0.9}{\ohm} in total.

\begin{figure}[!t]
    \centering
    \includegraphics[width=\columnwidth]{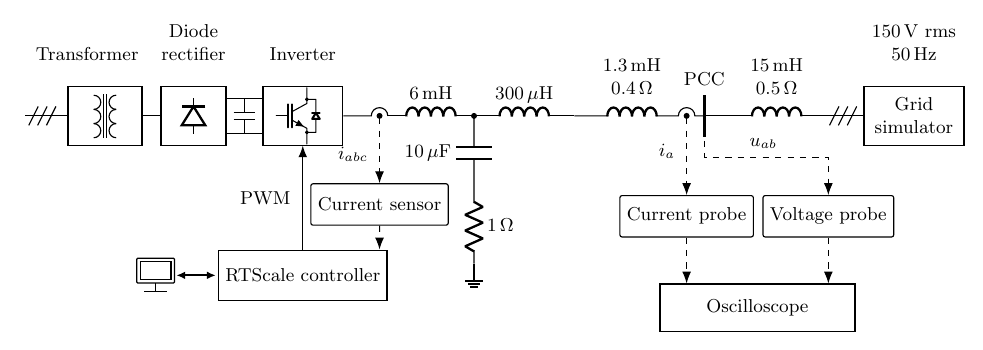}
    \caption{\color{black}Experimental circuit and PCC measurement connections
        for the line-to-line voltage $u_{ab}$ and phase current $i_a$.}
    \label{expe_circuit}
\end{figure}

The experimental results are reported in physical units, whereas the controller
is implemented in per unit with $U_b=150\sqrt{3}$ V and
$S_b=\SI{5.5}{\kilo\volt\ampere}$. The reported $|i|$ is the
power-invariant $dq$ current magnitude.
The OBS gains are fixed at
$k_o=0.1$, $k_p=0.02$, and $k_v=2$ throughout the experiments.

\begin{figure}[!t]
    \centering
    \includegraphics[width=\columnwidth]{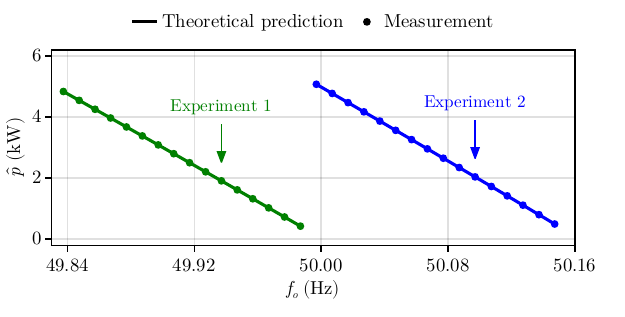}
    \caption{\color{black}Experimental OBS droop characteristics.}
    \label{expe_obs_frequency}
\end{figure}
\newcommand{\OBSProbePanel}[1]{%
  \begingroup
  \def\probeFile{pcc_frequency.PNG}%
  \def\probeScale{200 V/div; 5 A/div; 86 ms/div}%
  \def\probeStart{29.8}\def\probeEnd{49.2}\def\probeVoltageY{17.8}%
  \ifnum#1=2
    \def\probeFile{pcc_voltage_sag.PNG}%
    \def\probeScale{200 V/div; 5 A/div; 50 ms/div}%
    \def\probeStart{34.0}\def\probeEnd{67.3}%
  \fi
  \ifnum#1=3
    \def\probeFile{pcc_phase_jump.PNG}%
    \def\probeScale{200 V/div; 10 A/div; 50 ms/div}%
    \def\probeStart{15.8}\def\probeVoltageY{16.6}%
  \fi
  {\fontsize{7pt}{8pt}\selectfont\color{black}\probeScale\par}%
  \vspace{2pt}%
  \setlength{\unitlength}{0.01\linewidth}%
  \begin{picture}(100,40)
    \put(0,0){\includegraphics[width=\linewidth,height=0.35\linewidth,trim=40bp 44bp 40bp 100bp,clip]{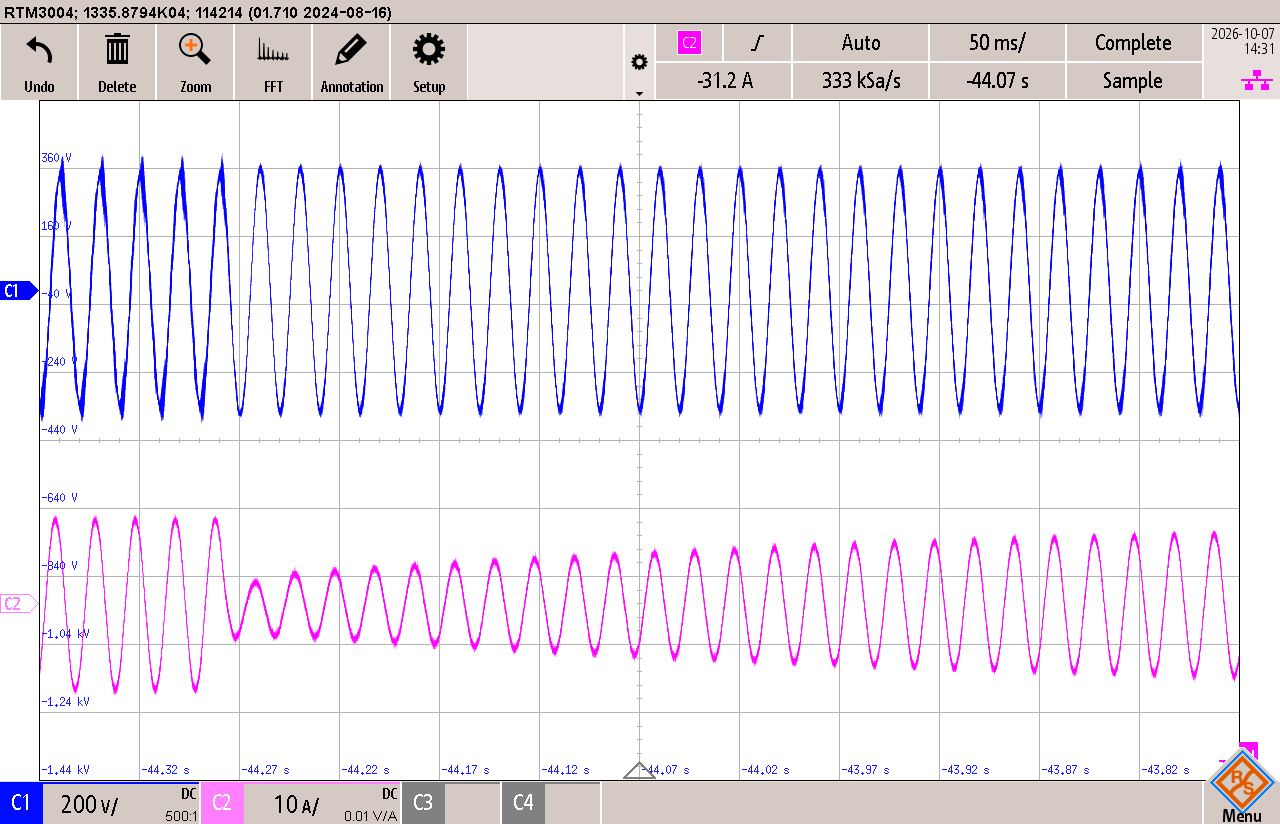}}%
    \ifnum#1=3
      \put(\probeStart,36.2){\makebox(0,0)[b]{\color{black}\fontsize{7pt}{8pt}\selectfont Phase jump}}%
    \else
      \put(\probeStart,36.2){\makebox(0,0)[b]{\color{black}\fontsize{7pt}{8pt}\selectfont Start}}%
      \put(\probeEnd,36.2){\makebox(0,0)[b]{\color{black}\fontsize{7pt}{8pt}\selectfont End}}%
      {\color{black}\linethickness{0.45pt}\multiput(\probeEnd,0)(0,1.36){26}{\line(0,1){0.75}}}%
    \fi
    {\color{black}\linethickness{0.45pt}\multiput(\probeStart,0)(0,1.36){26}{\line(0,1){0.75}}}%
    \put(96,\probeVoltageY){\makebox(0,0)[r]{\color{blue}\fontsize{7pt}{8pt}\selectfont$u_{ab}$}}%
    \put(96,3.0){\makebox(0,0)[r]{\color{magenta}\fontsize{7pt}{8pt}\selectfont$i_a$}}%
  \end{picture}%
  \endgroup
}

\begin{figure}[!t]
    \centering
    \includegraphics[width=\columnwidth]{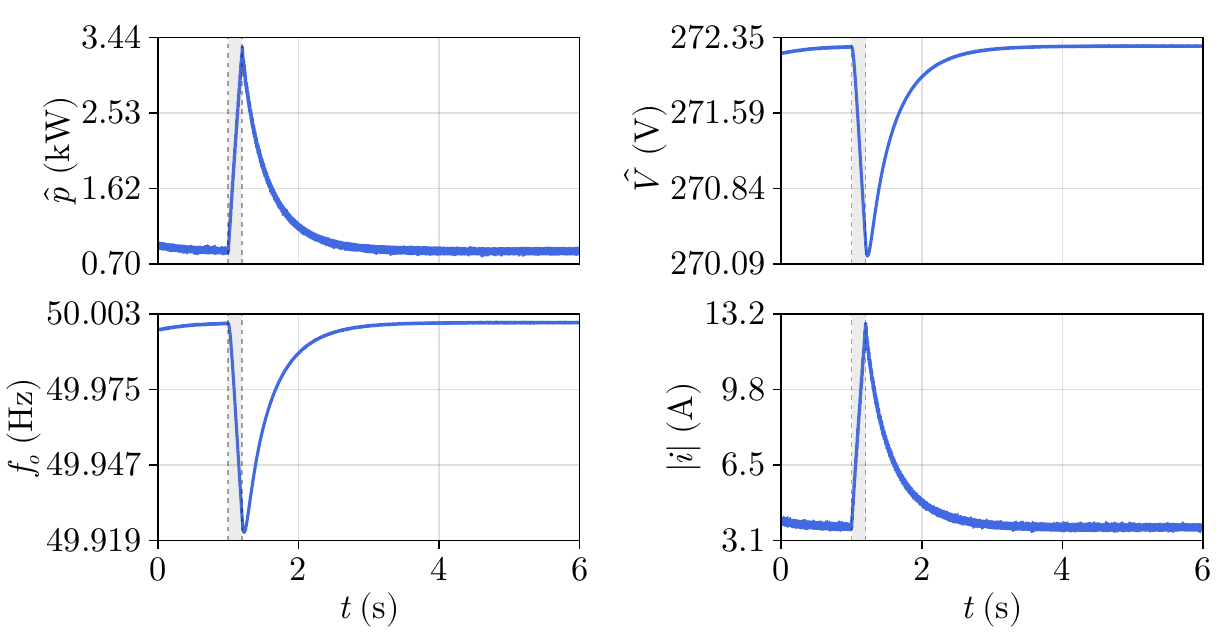}
    \par\vspace{2pt}
    {\footnotesize (a)\par}
    \par\vspace{5pt}
    \includegraphics[width=\columnwidth]{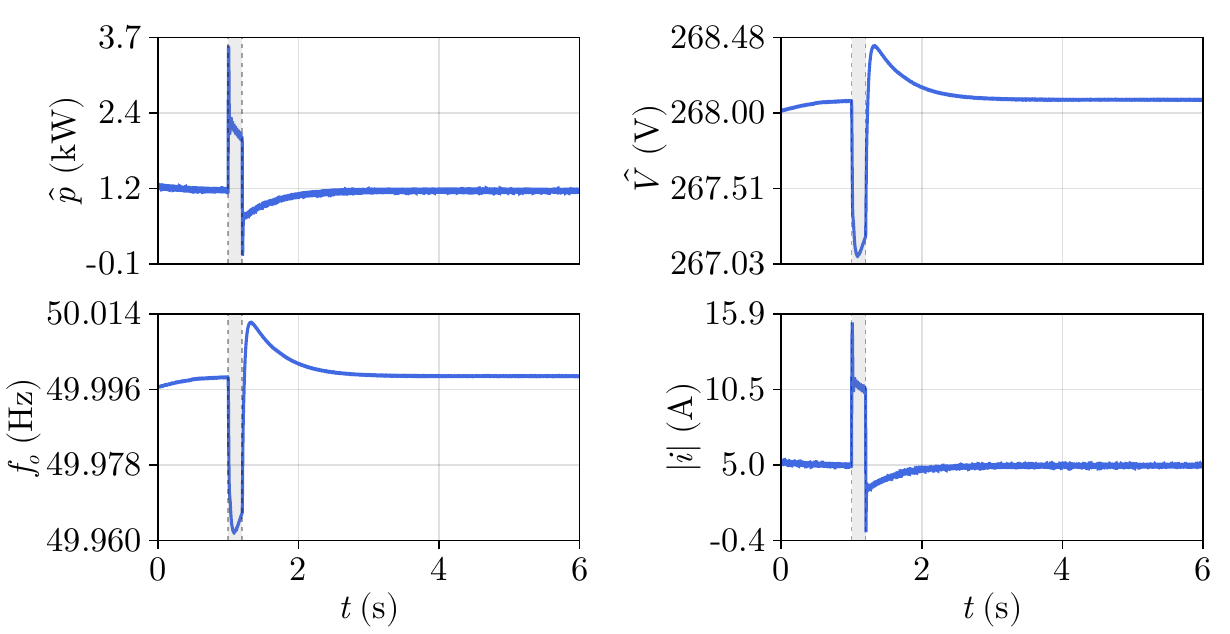}
    \par\vspace{2pt}
    {\footnotesize (b)\par}
    \par\vspace{5pt}
    \includegraphics[width=\columnwidth]{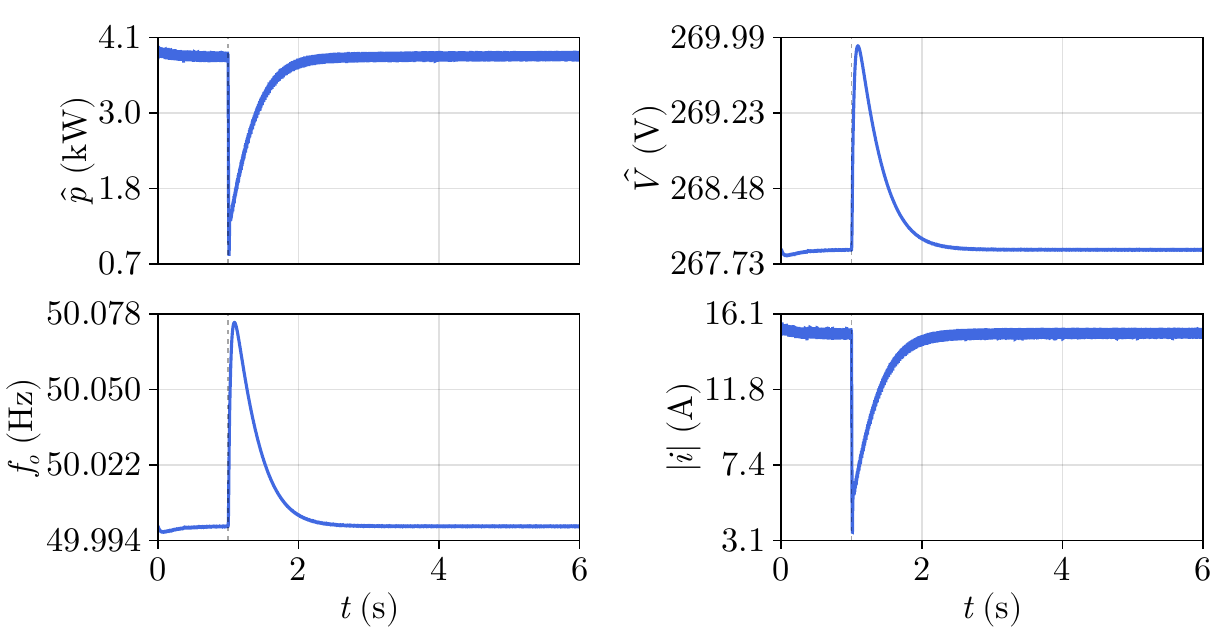}
    \par\vspace{2pt}
    {\footnotesize (c)\par}
    \caption{\color{black}OBS responses to (a) frequency reduction,
        (b) grid-voltage sag, and (c) phase jump.}
    \label{expe_obs_transients}
\end{figure}

\subsection{Nonlinear Droop Characteristics}
\label{subsec:experimental_steady_state}

\color{black}
Fig.~\ref{expe_obs_frequency} shows two sets of steady-state experimental data
obtained with different fixed references. The first set uses
$(p_r,v_r)=(\SI{4.98}{\kilo\watt},\SI{283.78}{\volt})$, while the second uses
$(p_r,v_r)=(\SI{0.03}{\kilo\watt},\SI{283.78}{\volt})$. Considering both
experimental records together, the applied grid-simulator frequency setpoints
cover an overall range from \SI{49.84}{\hertz} to \SI{50.15}{\hertz}. For each
frequency setpoint, a \SI{5}{\second} interval is selected after the transient
has settled. Each plotted point represents the time averages of $f_o$, $\hat V$,
and $\hat p$ over that interval.

The internal OBS frequency is calculated as
$f_o=(\omega_b/2\pi)(1+\Delta\omega_o)$ using
\eqref{eq:obs_internal_frequency}. In both records, $\hat V$ increases with
$f_o$, whereas $\hat p$ decreases. The maximum absolute deviations between the
measured $\hat p$ and this prediction are \SI{0.21}{\watt} and
\SI{0.43}{\watt} for the two records, respectively. Normalized by the rated power of \SI{5.5}{\kilo\watt}, the corresponding root mean square errors (RMSEs) are 0.00153\% and 0.00296\%, respectively. These results show close agreement
with the derived droop characteristic in \eqref{eq:steady_power_frequency}.

\color{black}
\begingroup\color{black}
\subsection{Responses to Grid Disturbances}
\label{subsec:experimental_disturbances}

The three disturbance experiments are recorded separately. Their time axes are
shifted so that each disturbance begins at $t=\SI{1}{\second}$ on the
common $0$--$6$ s axis.

Fig.~\ref{expe_obs_transients} shows the OBS responses to a
\SI{0.2}{\hertz} frequency reduction lasting \SI{0.2}{\second}, a grid voltage sag with the
phase-voltage RMS reduced from \SI{150}{\volt} to \SI{110}{\volt} lasting
\SI{0.2}{\second}, and a $10^\circ$ phase jump. The displayed quantities
are $\hat p$, $\hat V$, $f_o$, and $|i|$, arranged in two rows per case.

In Fig.~\ref{expe_obs_transients}(a), the frequency reduction raises
$\hat p$ from approximately \SI{0.86}{\kilo\watt} to
\SI{3.34}{\kilo\watt} and $|i|$ from \SI{3.79}{\ampere} to
\SI{12.85}{\ampere}, while $\hat V$ and $f_o$ decrease.
During the grid voltage sag in Fig.~\ref{expe_obs_transients}(b), $|i|$ reaches
\SI{15.33}{\ampere}. The phase-jump experiment in Fig.~\ref{expe_obs_transients}(c) starts
from an active-power level of approximately \SI{3.86}{\kilo\watt}.
The phase jump initially reduces $\hat p$ and $|i|$ to
\SI{0.82}{\kilo\watt} and \SI{3.55}{\ampere}, respectively, while
$\hat V$ and $f_o$ increase. In all three experiments, the responses
return to their respective pre-disturbance operating points.
The corresponding transients in the PCC line voltage $u_{ab}$ and phase
current $i_a$, measured using voltage and current probes, respectively,
are shown in the oscilloscope records in Fig.~\ref{expe_pcc_probes}.
\par\endgroup

\begin{figure}[!t]
    \centering
    \OBSProbePanel{1}
    \par\vspace{2pt}
    {\footnotesize (a)\par}
    \par\vspace{6pt}
    \OBSProbePanel{2}
    \par\vspace{2pt}
    {\footnotesize (b)\par}
    \par\vspace{6pt}
    \OBSProbePanel{3}
    \par\vspace{2pt}
    {\footnotesize (c)\par}
    \caption{\color{black}PCC $u_{ab}$ and $i_a$ waveforms during
        (a) frequency reduction, (b) grid-voltage sag, and (c) phase jump.}
    \label{expe_pcc_probes}
\end{figure}

\begingroup
\section{Conclusion}
\label{sec:conclusion}

\textcolor{black}{This paper develops a systematic nonlinear analysis of OBS control.
An exact nonlinear swing-equation representation characterizes the intrinsic inertia of OBS control. The nonlinear $p$--$\omega$ and $V$--$\omega$ characteristics and their dependence on the controller parameters are derived from the steady-state relations, providing guidelines for controller parameterization.
The swing-equation representation and droop characteristics also reveal
the structural relations of OBS control to other GFM synchronization laws.
For the reduced model, geometric and qualitative methods characterize
the equilibrium structure and nonlinear convergence properties, establishing
explicit sufficient parameter conditions for almost-global convergence.}

\textcolor{black}{EMT simulation results show that the reduced model accurately captures
the nonlinear dynamics of OBS control. These results are also consistent
with the theoretical convergence analysis.
Hardware experiments further verify
the derived droop characteristics and
demonstrate the synchronization capability of OBS control under grid
disturbances.}

\textcolor{black}{The proof of almost-global asymptotic stability is restricted to the reduced model of a single-inverter infinite-bus system. This model neglects fast dynamics, including power-network dynamics, and therefore does not fully capture the system's small-signal characteristics. Additionally, the OBS control considered here lacks $q$--$V$ droop for reactive-power sharing during parallel operation. In future work, we will extend the nonlinear analysis to multi-inverter systems, investigate the OBS damping mechanism with fast dynamics retained, and incorporate $q$--$V$ droop through an extended voltage correction.}
\par\endgroup

\FloatBarrier
\bibliographystyle{IEEEtran}
\begingroup
\expandafter\def\csname reviewerthreebibChandorkar1993Parallel\endcsname{}
\expandafter\def\csname reviewerthreebibTorres2015Synchronization\endcsname{}
\expandafter\def\csname reviewerthreebibColombino2019GlobalSync\endcsname{}
\expandafter\def\csname reviewerthreebibSeo2019DispatchableVOC\endcsname{}
\expandafter\def\csname reviewerthreebibHe2024ComplexDroop\endcsname{}
\expandafter\def\csname reviewerthreebibSimpsonPorco2013DroopSync\endcsname{}
\expandafter\def\csname reviewerthreebibJin2022DQFrame\endcsname{}
\expandafter\def\csname reviewerthreebibLiu2021SLVM\endcsname{}
\expandafter\def\csname reviewerthreebibTang2022ImprovedEAC\endcsname{}
\expandafter\def\csname reviewerthreebibYang2021Homoclinic\endcsname{}
\let\originalbibitem\bibitem
\renewcommand{\bibitem}[1]{%
    \ifcsname reviewerthreebib#1\endcsname\color{black}\else\color{black}\fi
    \originalbibitem{#1}%
}
\bibliography{ref}
\endgroup
\end{document}